\documentclass[journal,onecolumn,12pt]{IEEEtran}
\usepackage{amsthm,amssymb,graphicx,amsmath,color,amsfonts}
\usepackage[update,prepend]{epstopdf}
\usepackage[noadjust]{cite}
\usepackage[latin1]{inputenc}
\usepackage{pdfpages}
\usepackage{tabulary}
\usepackage{multirow}
\usepackage{bbm} 					% for \mathbbm{1}
\usepackage{comment}
\usepackage{dsfont}					% Fix type 3 fonts
\usepackage{subfigure}
\def\nbb{{\mathbf{b}}}

\def\nbn{{\mathbf{n}}}

\def\nbr{{\mathbf{r}}}

\def\nbv{{\mathbf{v}}}

\def\nbx{{\mathbf{x}}}

\def\nbzero{{\mathbf{0}}}
\def\nb1{{\mathbf{1}}}

\def\nbR{{\mathbf{R}}}

\def\nbV{{\mathbf{V}}}

\def\nbX{{\mathbf{X}}}

\def\ncalD{{\mathcal{D}}}

\def\ncalX{{\mathcal{X}}}

\def\nbbE{{\mathbb{E}}}

\def\nbbN{{\mathbb{N}}}

\def\nbbP{{\mathbb{P}}}

\def\nbbR{{\mathbb{R}}}

\def\nrmd{{\rm d}}

\newtheorem{lemma}{Lemma}

\newtheorem{theorem}{Theorem}
\newtheorem{prop}{Proposition}
\newtheorem{cor}{Corollary}

\newtheorem{remark}{Remark}

\def\figref#1{Fig.\,\ref{#1}}%

\def\chb#1{{\color{blue}#1}}
\begin{document}
\graphicspath{{./Figures/}}

\title{Foundations of Vision-Based Localization:
A New Approach to Localizability Analysis Using Stochastic Geometry}

\author{
	Haozhou Hu, Harpreet S. Dhillon, R. Michael Buehrer
	\thanks{The authors are with Wireless@VT, Bradley Department of Electrical and Computer Engineering, Virginia Tech, Blacksburg, VA, 24061, USA. Email: \{haozhouhu, hdhillon, rbuehrer\}@vt.edu. The support of the US NSF (Grant CNS-2107276 and CNS-2225511) is gratefully acknowledged. This work was also supported in part by the Commonwealth Cyber Initiative, an investment in the advancement of cyber R\&D, innovation, and workforce development. For more information about CCI, visit www.cyberinitiative.org.} 
}

% Part of this work has been accepted to Asilomar 2024~\cite{conf1}.

%The support of the US NSF (Grant CNS-2107276 and CNS-2225511) is gratefully acknowledged. This work was also supported in part by the Commonwealth Cyber Initiative, an investment in the advancement of cyber R\&D, innovation, and workforce development. For more information about CCI, visit www.cyberinitiative.org.

\maketitle

\begin{abstract}
Despite significant algorithmic advances in vision-based positioning, a comprehensive probabilistic framework to study its performance has remained unexplored. 
The main objective of this paper is to develop such a framework using ideas from stochastic geometry. 
Due to limitations in sensor resolution, the level of detail in prior information, and computational resources, we may not be able to differentiate between landmarks with similar appearances in the vision data, such as trees, lampposts, and bus stops. 
While one cannot accurately determine the absolute target position using a single indistinguishable landmark, obtaining an approximate position fix is possible if the target can see multiple landmarks whose geometric placement on the map is unique. Modeling the locations of these indistinguishable landmarks as a Poisson point process (PPP) $\Phi$ on $\mathbb{R}^2$, we develop a new approach to analyze the localizability in this setting. 
From the target location $\nbx$, the {\em measurements} are obtained from landmarks within the {\em visibility region}.
These measurements, including ranges and angles to the landmarks, denoted as $f(\nbx)$, can be treated as mappings from the target location.
We are interested in understanding the probability that the measurements $f(\nbx)$ are sufficiently distinct from the measurement $f(\nbx_0)$ at the given location, which we term localizability.
Expressions of localizability probability are derived for specific vision-inspired measurements, such as ranges to landmarks and snapshots of their locations. 
Our analysis reveals that the {\em localizability probability} approaches one when the landmark intensity tends to infinity, which means that error-free localization is achievable in this limiting regime.
\end{abstract}


\section{Introduction} \label{sec:intro}
\subsection{Problem Formulation}
Let $\Phi = \{\nbx_i, i\in \nbbN\}$ represent a PPP on $\nbbR^2$. For any $\nbx \in \nbbR^2$, we introduce the function $g(\nbx) = \Phi_{-\nbx} \cap \nbb(\nbzero, d)$, defined as the intersection of the translated point process $\Phi_{-\nbx} = \{\nbx_i - \nbx, i \in \nbbN\}$ and the ball $\nbb(\nbzero, d)$ centered at the origin with radius $d$. 
Subsequently, we define a function composition $f(\nbx) = h \circ g(\nbx)$, where $h$ is a mapping from the intersection of the translated point process $g(\nbx)$ to its codomain.
The distance between two elements in the codomain of $f$ is represented as $\Delta(f(\nbx_i), f(\nbx_j))$.
For a specified location $\nbx_0$ on $\nbbR^2$, its corresponding $f(\nbx_0)$ denotes the signature of $\nbx_0$.
Our interest lies in determining the probability that the signature of an arbitrary location $f(\nbx)$ is within $\epsilon$ distance of $f(\nbx_0)$.
The motivation and physical significance behind this problem are explained within the framework of vision-based positioning next.

\begin{comment}
    \chb{Let $\Phi = \{\nbX_i\}$ be a PPP on $\nbbR^2$. The geometric pattern of 
a given location $\nbx$ is the points within the ball $\nbb(\nbx, d)$ centered at $\nbx$ with radius $d$, defined as $g(\nbx) = \nbb(\nbx, d) \cap \Phi$.
The measurement is derived from the geometric pattern at location $\nbx$ and is expressed as a function of its location, denoted by $f(\nbx) = h \circ g(\nbx)$.
Measurements from two locations are considered similar if the distance between them, $d(f(\nbx_1), f(\nbx_2))$, is less than or equal to some threshold $\epsilon$. 
Given a realization of the point process $\varphi \in \Phi$ and measurements at a specific location $\nbx_0$, the probability of finding similar measurements elsewhere on the map is represented as
\begin{align}
    \nbbP\left[d(f(\nbx_0), f(\nbX)) \leq \epsilon \mid \Phi\right],
\end{align}
where $\nbX$ is an arbitrarily chosen location on the map.
Further, with a given realization of PPP, the probability that measurements from two arbitrarily chosen locations are similar is of interest and defined as 
\begin{align}
    \nbbP\left[d(f(\nbX_0), f(\nbX)) \leq \epsilon \mid \Phi\right].
\end{align}
}
\end{comment}

\subsection{Background}
Localization involves determining the position and orientation of the {\em target} using various measurements. In vision-based approaches, these measurements consist of visual information captured by vision sensors. These approaches are widely used in robotics, autonomous vehicles, augmented reality, and various other fields.
However, the foundations of vision-based positioning have not been thoroughly explored to the same depth as those of wireless-based positioning because of two interrelated challenges.
First, unlike wireless signals that can be uniquely associated with their corresponding anchors, the limited resolution of vision sensors blurs the distinctions between geographical landmarks (such as trees) with similar appearances. As a result, even when a landmark is identified in the visual data, it may not necessarily aid in accurate positioning, given the presence of multiple indistinguishable landmarks in a region. 
Second, wireless-based positioning benefits from well-established mathematical principles drawn from wireless communications, information theory, and estimation theory, enabling rigorous analysis. In contrast, similar formal mathematical frameworks for vision-based positioning have not been as extensively developed. This paper introduces one such mathematical framework for vision-based positioning utilizing tools from stochastic geometry, a natural choice to describe random spatial patterns.

Specifically, we model the map of landmarks as a PPP $\Phi = \{\nbx_i, i \in \nbbN\}$ on $\nbbR^2$.
The map of landmarks and the measurements of the visible landmarks are utilized to estimate the unknown location.
The measurements obtained at a given location can be treated as mappings, including ranges, angles, and relative locations, of the proximate landmarks as seen from that location. 
%These measurements are generally represented as mappings of their corresponding locations, including ranges, angles, and relative locations to the visible landmarks.
If the obtained measurements are {\em unique} enough, they effectively assist in positioning.
To quantify the ability of measurements to identify locations, we investigate the probability that measurements from two different places are similar.
This property implies the {\em uniqueness} of the measurements and the amount of information provided about the location.
{\em It should be noted that this perspective of vision-based positioning, probabilistic problem formulation, and the subsequent application of stochastic geometry analysis are all presented for the first time in this article.} This approach is of interest to the statistical physics community because of the use of point processes to model landmark locations.

\subsection{Related Literature and Contributions}
In this article, we explore the new connections between vision-based localization and stochastic geometry. The relevant prior arts from both directions are discussed. We also briefly discuss the connections to localizability-related works and information theory.

{\em Vision-based localization.} Traditionally, vision-based positioning has been approached as an image retrieval problem that estimates the unknown location of a query image by cross-referencing it with the most similar geo-tagged images in a pre-established database~\cite{zheng2009tour,crandall2009mapping,li2008modeling,zamir2016introduction}. 
Within this framework, SIFT (Scale-Invariant Feature Transform) emerges as a prominent technique, known for its ability to extract transform-invariant features from images~\cite{lowe2004distinctive}.
In the matching process, images are represented by vectors, which are computationally more efficient than pixel-based representations when matching images in the database.
%In the matching process, instead of representing an image by pixels, a vector representation is used, which is more computationally efficient when matching against the database for positioning. 
These vectors are transform-invariant features, including but not limited to bag-of-words~\cite{csurka2004visual, philbin2007object}, VLAD~\cite{arandjelovic2013all}, and graph-based methodologies~\cite{shankar2016approximate, verde2020ground}.
Another popular approach in vision-based positioning is Perspective-n-Points (PnP). It determines the position and orientation of the camera by utilizing a set of 3D points and their corresponding 2D projections onto the image plane. 
Distinctive feature points across different views, such as corners, edges, or blobs where the image intensity changes significantly, are detected. 
The locations of these feature points are used to estimate the camera location by solving parameters in the camera model~\cite{gao2021introduction}.
Other approaches involve convolutional neural networks (CNNs) to extract features from images. They have shown improved performance, particularly in various challenging environments.
Several variants of CNN, such as PoseNet~\cite{kendall2015posenet}, MapNet~\cite{geometry-aware}, and CamNet~\cite{ding2019camnet}, have been developed to locate and track the position and the orientation of the camera. 
The performance of these models is primarily evaluated on small-scale datasets, such as 7-Scenes and Oxford RobotCar.
In parallel, large-scale localization is explored using {\em cross-view matching}, which aims to match ground-based imagery with aerial images (such as from the satellites). 
Several neural network structures have been proposed for positioning the image~\cite{zhu2021revisiting, zhu2021vigor, hu2022beyond}.

{\em Stochastic Geometry and Statistical Physics.} We will now discuss relevant prior art from stochastic geometry. Even though no prior works use stochastic geometry for vision-based positioning (except our conference presentations~\cite{conf1,conf2} reporting some preliminary results), it has been used extensively in wireless communications and positioning. For instance, a recent approach to modeling the locations of anchors and blockages in wireless-based positioning involves treating them as realizations of point processes. In this context, stochastic geometry provides the framework to analyze key localization metrics, including localizability, e.g., see~\cite{o'lone2018statistical, christopher2019mathematical, schloemann1, schloemann2,aditya2018tractable,aditya2019characterizing, elsawy2017base}. These analyses help identify key factors in localization, evaluate the impact of these factors on the network performance, and suggest guidelines for optimizing localization algorithms. Point processes and stochastic geometry in this line of work also offer a common thread with statistical physics, where the same tools have found numerous applications. To explain this connection rigorously, consider a popular model for wireless cellular networks in which the locations of mobile towers are modeled as a PPP~\cite{andrews2023introduction}. Under realistic assumptions on how users connect to each mobile tower, the service regions of these towers can be modeled as Poisson Voronoi cells with mobile towers at the nuclei. 
To determine the transmit power of each user, it is important to understand the distances between users and their respective associated towers.
The distribution of this distance was recently derived in~\cite{mankar2020distance} using techniques developed in statistical physics~\cite{pineda2007temporal} to study the temporal evolution of the domain structure of a Poisson Voronoi tessellation. The moments of the volume of the cells located at the edges of a bounded Poisson Voronoi tessellation were derived in~\cite{koufos2019distribution}. Along the same lines, resource allocation problems in wireless networks have inspired a new variant of a random sequential adsorption process, termed the multilayer random sequential adsorption process, in~\cite{parida2022multilayer}. Another common thread is the use of line processes in both wireless and statistical physics. In wireless, it has been recently used to model the underlying road systems in vehicular networks~\cite{dhillon2020poisson}, whereas, in statistical physics, it has found applications in modeling the trajectories of subatomic particles. The recent use of line processes in vehicular networks has also inspired a new line of questioning related to path distances (distances along the lines), which has been studied in statistical physics~\cite{chetlur2020shortest}.

{\em Localizability.} Localizability analysis in wireless-based positioning systems provides initial estimations of the localization error by focusing on the feasibility of obtaining a position fix in a given setting. For instance, the Cram\'{e}r Rao Lower Bound related to positioning estimation is utilized to evaluate system settings, such as the configurations of anchor locations, network throughput, and estimation methods for localization~\cite{crlb1, crlb3, crlb2, o'lone2018statistical}. In~\cite{christopher2019mathematical,christopher2019single}, the localization performance using non-line-of-sight (NLOS) path reflections is investigated by analyzing localizability. However, the localizability in vision-based positioning has not been thoroughly explored to the same depth as in wireless-based positioning. Existing works either lack a rigorous mathematical model or are limited to specific applications.
In~\cite{li2019localizability}, localizability is evaluated through location entropy, estimated using a joint neural network. Another work in Lidar simultaneous localization and mapping (SLAM) 
dynamically adjusts the matching parameters based on localizability to improve localization performance~\cite{dong2021efficient}.

{\em Information Theory.} Another relevant research direction, though less immediately evident, is the information theoretic exploration of point process models.
For instance, \cite{anantharam2015capacity} studies capacity and error exponents of stationary point processes by considering points in the process as codewords and random displacement of points as additive noise.
This work shows that error-free communication is achievable when the rate does not exceed the Poltyrev capacity.
Further, the entropy rate of stationary point processes and the mutual information between a Poisson
point process and its displaced version are derived in~\cite{baccelli2016entropy}.
Even though one can draw parallels between our work and this specific line of work (e.g., the measurements in our model are analogous to codewords), the line of questioning (inspired by vision-based localization) and mathematical development are fundamentally different. 
That said, the problem discussed in this paper has similarities to the communication process.
The target's position fix can be viewed as a source message encoded by its corresponding error-free measurements. The process of obtaining the measurements can be considered as error-free measurements passing through the noisy channel. 
At the receiver, the received signal is a noise-affected measurement.
While the goal of the communication system is to retrieve the source message, the objective of localization is to estimate the position fix using these noise-affected measurements. 
%\chr{As the channel capacity establishes the fundamental limits of communications, the fundamental limits of localization systems are valuable to explore within the framework of information theory.}

{\em Contribution.} Even though stochastic geometry and statistical physics models have been used to study various aspects of wireless and positioning systems, no such attempt has been made to develop similar mathematical frameworks to study the localizability of vision-based systems. 
The challenge is that landmarks often have similar appearances and are indistinguishable in vision data, which is fundamentally different from processing signals in wireless-based positioning.
This paper proposes a fundamentally new perspective to derive localizability probability for vision-based positioning. We specifically focus on situations where all landmarks are indistinguishable in the vision data. 
%We consider the worst scenario that all landmarks are indistinguishable in the vision data. 
The landmark locations are modeled as a PPP, and the measurements are obtained from the landmarks that are visible to the target.
To make our framework more general, we represent the measurements as mappings from the locations where they are obtained. 
Specifically, the measurement is treated as a function of visible landmarks, including ranges, angles, and relative locations to these landmarks.
Modeling the measurement error as additive noise, 
we derive the expressions of the localizability probability
in vision-based positioning for a variety of measurements, including the number of visible landmarks, ranges, and relative locations to the visible landmarks.
The localizability probability quantifies the probability that the measurements obtained at one location match those obtained at other places on the map.
A lower localizability probability indicates that the obtained measurement is more uniquely associated with its location, making it easier to determine a global position fix.
Our analysis yields some interesting observations and connections to information theory. 
For instance, when the landmark intensity tends to infinity, the localizability probability will tend to one. 
This observation parallels the concept of having codewords with infinite lengths in communication systems. 
The implication is that error-free localization is achievable in this limiting scenario.

\section{Model and Problem Formulation} \label{sec:model_prob}
The landmarks that appear in the visual data are treated as points. We simplify the representation of the environment and focus on the spatial relationships between these points for localization purposes.
The locations of these points are modeled as a homogeneous PPP $\Phi = \left\{\nbx_i, i \in \nbbN\right\} \subset \nbbR^2$, where $\nbx_i$ denotes the landmark location. 
The landmark is {\em visible} to the target if the distance between them is less than the {\em maximum visibility distance} $d_v$. Thus, landmarks within the ball $B_{\nbx} = \nbb\!\left( \nbx, d_v\right) $ are {\em visible} to the target, where $\nbx$ is the target location. 
{\em It is worth noting that the landmark locations have been shown to follow a PPP in some settings of interest to vision-based localization, such as in~\cite{rohde2016localization}.}
The point representation of landmarks in this paper also has conceptual similarities with the graphical representation of landmarks discussed in prior works~\cite{verde2020ground}. In addition to providing a complementary probabilistic approach to this problem, our work offers an alternate definition for the weights of edges in these graphs through the measurement from vision sensors. Since we focus on probabilistic analysis, we will not explore this graph theory perspective in this paper. 

\subsection{Landmark Patterns} 
The visual information surrounding the target and the map of landmark locations are used to localize the target.
Specifically, the target location is inferred by comparing the actual geometric arrangement of landmarks and the locations on the map.
The geometric arrangement of landmarks around location $\nbx$ is defined as the {\em landmark pattern} at $\nbx$, represented as
\begin{equation} \label{eqdef:lmk-pattern}
	g(\nbx) = \left\{\tilde{\nbx}_i \mid \tilde{\nbx}_i = \nbx_i - \nbx, \nbx_i \in \Phi \cap B_{\nbx}\right\},
\end{equation}
where $g$ is a mapping from $\nbx$ to the relative locations of visible landmarks. 
It is important to note that there is no absolute location reference, and the absolute locations of these visible landmarks remain unknown.

\subsection{Measurements}
The form of the measurements obtained by the target varies depending on the application. Examples of the {\em measurements} include ranges, angles, and the relative locations of visible landmarks. 
To maintain generality, we model the measurement as a function of its location, denoted by $f(\nbx)$, where $\nbx$ represents the location where the measurement is obtained. 
The function $f = h \circ g$ is a function composition, where the mapping $h$ represents the {\em measurement process} on the landmark pattern. 
The measurement error is modeled as additive noise with finite support, meaning all possible noisy measurements at a location are bounded. We use a bounded set $F_{\epsilon}(\nbx)$ to represent all possible noisy measurements that could be obtained at $\nbx$.

\begin{comment}
\chb{Depending upon the application, the target located at $\nbx$ might take different measurements on the landmark pattern $g(\nbx)$.
%Depending upon the application, the target might take different measurements on $g(\nbx)$, which means that the measurements related to $g(\nbx)$ will be application-specific. 
Examples of the {\em measurements} include depth/range, relative angles, and the relative locations of visible landmarks. 
%For our analysis, describing specific methods for obtaining these measurements is unimportant. 
To maintain generality, we model the measurement as a function of its location, denoted by $f(\nbx)$, where $\nbx$ represents the target location, and $f = h \circ g$ is a function composition. 
The mapping $h$ represents the {\em measurement process} and is application-specific. 
The measurement error is modeled as additive noise with finite support.
Consequently, the measurements obtained at location $\nbx$ are random and take values within the set $F_{\epsilon}(\nbx)$.}
\end{comment}

\subsection{Problem Formulation}
\begin{figure} \label{fig:system-model}
	\centering
	\includegraphics[width=0.7\textwidth]{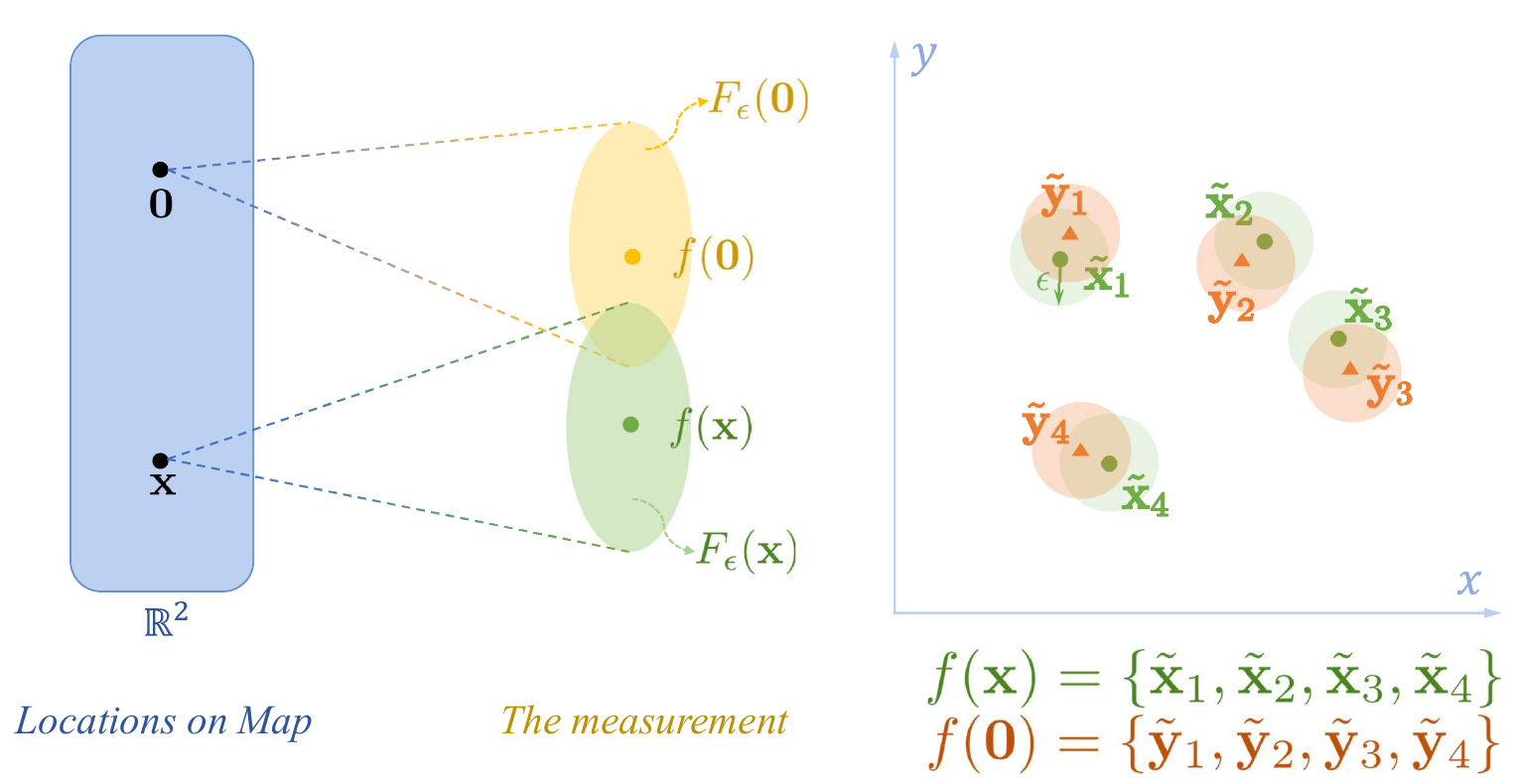}
	\caption{An illustration of the {\em overlapping} noisy measurements.}
\end{figure}
The localization performance is evaluated at candidate locations on the map.
These candidate locations form a PPP $\Psi = \left\{\nbx_i, i \in \nbbN\right\} \subset \nbbR^2$.
Conditioned on a point $\nbx_0 \in \Psi$, the probability that $F_{\epsilon}$ and $F_{\epsilon}(\nbx_0)$ do not overlap, is donated as
\begin{align} 
    \nbbP^{\nbx_0}\left[F_{\epsilon}(\nbx_0) \cap F_{\epsilon} = \varnothing\right] = \nbbP\!\left[ F_{\epsilon}(\nbx_0) \cap F_{\epsilon} = \varnothing \mid \nbx_0 \in \Psi \right],
\end{align}
where $\nbbP^{\nbx_0}(\cdot)$ is termed the Palm probability conditioned on a point at $\nbx_0$~\cite{chiu2013stochastic,BacBlaB2009,andrews2023introduction}, $F_{\epsilon}$ is a given bounded set of measurements, $F_{\epsilon}(\nbx_0)$ is the set of all possible measurements at $\nbx_0$. Palm probability characterizes the properties of a point process when one of its points is conditioned to lie at a specific location, which in this case is $\nbx_0$. 
For stationary point processes, the exact choice of $\nbx_0$ does not matter and hence we can consider it to be the origin \begin{align}
    \nbbP^{\nbx_0}\left[F_{\epsilon}(\nbx_0) \cap F_{\epsilon} = \varnothing\right] = \nbbP^{\nbzero}\left[F_{\epsilon}(\nbzero) \cap F_{\epsilon} = \varnothing\right].
\end{align}
An interpretation of the Palm probability is the proportion of points with the desired property in an arbitrary Borel set $B$, given by
\begin{align}
    \nbbP^{\nbzero}\left[F_{\epsilon}(\nbzero) \cap F_{\epsilon} = \varnothing\right] &= \frac{1}{\lambda |B|} \nbbE \left[ \sum_{\nbx \in \Psi} \nb1\!\left(F_{\epsilon}(\nbx) \cap F_{\epsilon} = \varnothing\right) \nb1\!\left(\nbx \in B\right) \right].
\end{align}
In terms of localization, it is the proportion of candidate locations where the obtained measurements are not in $F_{\epsilon}$, reflecting how often the measurements in $F_{\epsilon}$ appear on the map. 
Now, we apply Slivnyak's theorem~\cite{chiu2013stochastic,BacBlaB2009,andrews2023introduction} and define the {\em conditional localizability probability}
\begin{align}
    P_{\rm C, Loc} = \nbbP^{\nbzero}\left[F_{\epsilon}(\nbzero) \cap F_{\epsilon} = \varnothing\right] = \nbbP\left[F_{\epsilon}(\nbzero) \cap F_{\epsilon} = \varnothing\right].
\end{align}
This essentially follows from the fact that in a PPP, conditioning on a point at the origin does not change the distribution of the rest of the point process.
The {\em conditional localizability probability} represents the probability that the obtained measurements at the origin are not in $F_{\epsilon}$.

Next, we define the bounded set $F_{\epsilon}$ as the noisy measurements obtained from an arbitrarily selected candidate location $\nbx$. 
We are interested in the probability that the measurements obtained at the origin and the selected candidate location do not overlap. 
The probability of this event is defined as the {\em localizability probability}
\begin{align} \label{eqdef:Ploc}
    P_{\rm Loc} = \nbbP\left[F_{\epsilon}(\nbzero) \cap F_{\epsilon}(\nbx) = \varnothing\right],
\end{align}
where $F_{\epsilon}(\nbzero)$ and $F_{\epsilon}(\nbx)$ represent the noisy measurements obtained at the origin and $\nbx$, respectively. 
Notably, in~\eqref{eqdef:Ploc}, when the location $\nbx$ is positioned near the origin, the measurements will be similar to those obtained exactly at the origin. 
However, this does not present any technical issues, as clarified in the following remark
\begin{remark} \label{rem:localizability}
	Intuitively, measurements obtained around the origin will also be similar to those obtained strictly at the origin. 
    Even if candidate locations around the origin (which are finite) yield similar measurements, the number of such locations remains bounded. This implies that the probability of selecting a candidate location from around the origin is almost surely zero. 
    As a result, additional constraints are not needed in our analysis to capture the above intuition about localizability.
\end{remark}
Since the landmark locations form a homogeneous PPP, the landmarks around the origin and $\nbx$ can be considered independently and identically distributed.
Therefore, the probability in~\eqref{eqdef:Ploc} can be evaluated using the measurements from two independent realizations of the landmark locations. 
For convenience, we will occasionally present results for the {\em non-localizability probability} defined as
\begin{equation} \label{eqdef:nPloc}
	P_{\rm N-Loc} = \nbbP\!\left[F_{\epsilon}(\nbzero) \cap F_{\epsilon}(\nbx) \neq \varnothing\right] = 1 - P_{\rm Loc}.
\end{equation}

From an information theory perspective, we can treat the target location $\nbx$ as the source message. The measurements at location $\nbx$ encode this message, represented by the codeword $f(\nbx)$. 
The measurement error can be considered as additive noise in the channel. The receiver receives a noise-distorted codeword $\tilde{f}(\nbx)$ and aims to recover the source message $\nbx$.
Two types of errors occur when decoding the source message: (a) multiple locations are encoded with the same codewords, meaning that measurements at different locations are identical, and (b) the noise-distorted codeword lies in the decision region associated with another source message.
The former type of error suggests that the information conveyed by the codeword $f(\nbx)$ is insufficient to determine the location $\nbx$. 
For example, if the measurement is the number of visible landmarks, denoted as $f(\nbx) = | g(\nbx)|$, other locations on the map may also have the same number of visible landmarks as $\nbx$.
The latter type of error relates to the design of codewords in the presence of noise.
Ensuring that codewords (or measurements) are selected to provide more reliable and robust information becomes important.
Specific connections between our model and information theory are explored further in later sections.

\section{Number of Visible Landmarks} \label{sec:num_lmk}
Before we delve into the various types of measurements, we start with the most straightforward measurement: the number of visible landmarks. 
Here, the visible landmarks refer to the landmarks within the visibility region of the target.
The corresponding mapping $f$ is the cardinality of the landmark pattern $g(\nbx)$, denoted as:
\begin{equation}
	f: \nbbR^2 \mapsto \nbbN, \nbx \mapsto f(\nbx)= | g(\nbx) | = N\!\left(B_{\nbx}\right) = M_{\nbx},
\end{equation}
where $M_{\nbx}$ represents the number of visible landmarks at location $\nbx$. 
The non-localizability probability in~\eqref{eqdef:nPloc} can be written as
\begin{align}
	P_{\rm N-Loc} = \nbbP\!\left[F_{\epsilon}(\nbx) \cap F_{\epsilon}(\nbzero) \neq \varnothing \right] = \nbbP\!\left[M_{\nbx} = M_{\nbzero}  \right],
\end{align}
which is the probability that the numbers of visible landmarks at location $\nbx$ and the origin are the same.
We derive the localizability probability based on the number of visible landmarks in the following lemma.
\begin{lemma} \label{lem1}
	The localizability probability $P_{\rm Loc}$ based on the number of visible landmarks is
	\begin{align} \label{lem1-lb}
		P_{\rm Loc} = 1 - e^{-2 m} I_0 (2 m),
	\end{align}
	where $m = \lambda |B_{\nbx}| = \lambda \pi d_v^2$ is the average number of visible landmarks.
\end{lemma}
\begin{proof}
	Since $M_{\nbx}$ and $M_{\nbzero}$ are two independent Poisson random variables, the probability of $M_{\nbx} = M_{\nbzero}$ can be written as 
	\begin{align}
		\nbbP\!\left[M_{\nbx} = M_{\nbzero} \right] &= \sum_{k = 0}^{\infty} \nbbP\!\left[M_{\nbzero} = k \right] \cdot \nbbP\!\left[M_{\nbx} = k \right] \label{lb}\\
		&= \sum_{k = 0}^{\infty} \left(\frac{m^k}{k!} e^{-m}\right)^2 \\
		&= e^{-2m} I_0\!\left(2m\right),
	\end{align}
	where $I_0(\cdot)$ is the modified Bessel function of the first kind. Therefore, the localizability probability is 
	\begin{equation} \label{lb-proof}
		P_{\rm Loc} = 1 - \nbbP\!\left[M_{\nbx} = M_{\nbzero} \right] = 1 - e^{-2m} I_0\!\left(2m\right).
	\end{equation}
	This completes the proof.
\end{proof}
Lemma~\ref{lem1} presents the preliminary result on the localizability probability based on the number of visible landmarks. It illustrates that $P_{\rm Loc}$ is an increasing function of $m$, meaning that a larger visibility distance $d_v$ or landmark intensity $\lambda$ will lead to a higher localizability probability.
Additionally, $P_{\rm Loc}$ established in Lemma~\ref{lem1} serves as a lower bound for the localizability probabilities based on other types of measurements, as examined in the subsequent sections.
A straightforward explanation is that other types of measurements incorporate additional information. 
For example, the information conveyed by the number of landmarks is contained within the number of ranges and relative locations of landmarks.

\begin{remark} \label{remark}
    Assuming the target can always determine its location whenever a visible landmark is present, an upper bound of the localizability probability can be obtained by considering the $k=0$ term in equation~\eqref{ub}
	\begin{equation}
		\label{ub}
		P_{\rm Loc} \leq 1 - e^{-2m}.
	\end{equation}
    When there is no visible landmark, no additional information can be obtained beyond the number of visible landmarks being zero.
\end{remark}

\section{Case 1: Range Vectors}\label{sec:case1}
In this section, we consider the scenario where the target is equipped with a range scanner that can sequentially measure the ranges to each visible landmark, starting from a fixed orientation, such as true north.
The ordering of the range measurements provides information about the geometric arrangement of the landmarks. 
The obtained range measurements at location $\nbx$ can be represented as a vector, given by
\begin{equation}
	\nbr_{\nbx} = \left[r_1, \dots, r_i, \dots, r_{N_\nbx}\right],
\end{equation}
where $r_i = \|\tilde{\nbx}_i\|, \tilde{\nbx}_i \in g(\nbx)$ is the range to the $i$-th landmark, counted clockwise from the true north. Thus, the mapping $f$ is
\begin{equation}
	\label{def:c1-f}
	f: \nbbR^2 \mapsto \nbbR^{N_\nbx}, \nbx \mapsto f(\nbx) = \nbr_{\nbx},
\end{equation}
where $\nbr_{\nbx}$ represents the {\em range vector} at the location $\nbx$. We assume that measurement error associated with each range measurement $r_i \in \nbr_{\nbx}$ has finite support and does not exceed $\epsilon/2$.
Therefore, the possible outcomes of the noisy measurements at $\nbx$ can be characterized using a bounded set
\begin{equation}
	\label{def:c1-fe}
	F_{\epsilon}(\nbx) = \left\{\tilde{\nbr}_{\nbx} \mid \tilde{\nbr}_{\nbx} = \nbr_{\nbx} + \nbn, \| \nbn \|_{\infty} \leq \epsilon/2 \right\}.
\end{equation}
We define the distance between range vectors by generalizing the Chebyshev distance
\begin{equation}
	\label{def:c1-vec-dist}
	\Delta_{\rm p}\!\left(\nbr_i, \nbr_j\right) = 
	\begin{cases}
		\left\| \nbr_i - \nbr_j \right\|_{\infty}, & \dim(\nbr_i) = \dim(\nbr_j)\\
		\infty, & \dim(\nbr_i) \neq \dim(\nbr_j)
	\end{cases},
\end{equation}
where $\left\|\cdot\right\|_{\infty}$ is the infinity norm. Now, the definition in~\eqref{def:c1-fe} is equivalent to
\begin{equation}
	F_{\epsilon}(\nbx) = \left\{\tilde{\nbr}\mid \Delta_{\rm p}(\tilde{\nbr},\nbr_{\nbx}) \leq \epsilon/2 \right\}.
\end{equation}
\begin{figure*}
	\centering
	\subfigure[]{
	\begin{minipage}[t]{0.3\textwidth}
		\centering
		\includegraphics[width=\textwidth]{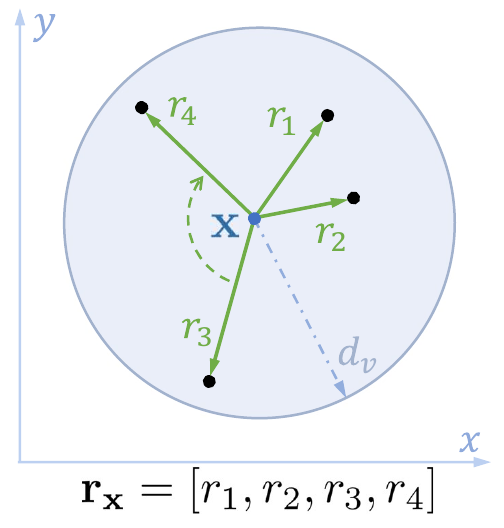}
		\label{fig:c1-model-1}
	\end{minipage}
	}
	\subfigure[]{
	\begin{minipage}[t]{0.3\textwidth}
		\centering
		\includegraphics[width=\textwidth]{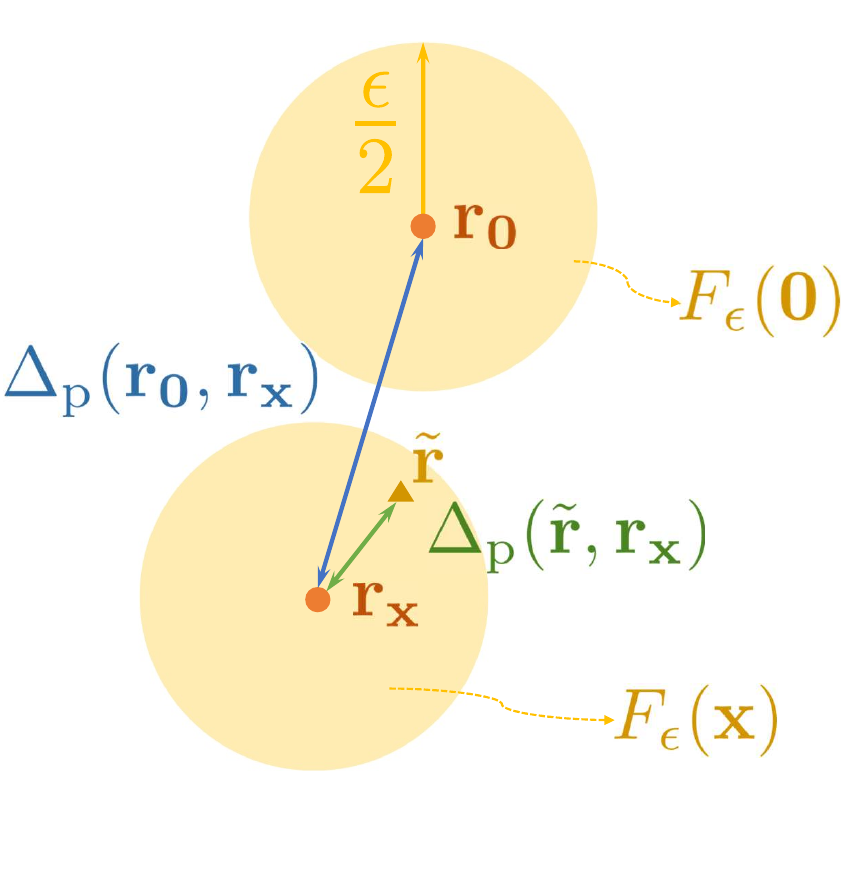}
		\label{fig:c1-model-2}
	\end{minipage}
	}
	\caption{An illustration of (a) the range vector measurement and (b) the distances between range vector measurements.}
	\label{fig:c1-model}
\end{figure*}

As we show in~\figref{fig:c1-model-2}, $\nbr_\nbzero$ and $\nbr_\nbx$ represent the noise-free range vectors obtained at the origin and at $\nbx$, respectively. When $\Delta_{\rm p}\!\left(\nbr_{\nbzero}, \nbr_{\nbx}\right) > \epsilon$, the noisy measurements at two locations do not overlap. 
Therefore, the localizability probability is equivalently written as
\begin{align}
	\label{def:c1-ploc}
	P_{\rm Loc} = \nbbP\! \left[\Delta_{\rm p}\!\left(\nbR_{\nbzero}, \nbR_{\nbx}\right) > \epsilon\right],
\end{align}
which represents the probability that the noise-free ordered distance vectors at locations $\nbx$ and the origin have a minimum separation of $\epsilon$. We present the statistical results necessary for analyzing localizability probability in the following lemmas.
First, given the number of visible landmarks, we derive the distribution of range measurements in the following lemma.
\begin{lemma}
	\label{lem2}
	Given a fixed number of visible landmarks $N(B_{\nbx}) = N_\nbx = k , (k > 0)$, the conditional distribution of the distance $R$ from any visible landmark to the location $\nbx$ is 
	\begin{equation}
		f_R (r) = \frac{2r}{d_v^2} \cdot \delta\!\left(0 \leq r \leq d_v\right),
	\end{equation}
	where $\delta\!\left(0 \leq r \leq d_v\right)$ is the Dirichlet function
    \begin{align}
     \delta\!\left(0 \leq r \leq d_v\right) = \begin{cases}
         1, &0 \leq r \leq d_v \\
         0, &otherwise
     \end{cases},
    \end{align}
    and $\nbx$ is an arbitrary location on the map.
\end{lemma}
\begin{proof}
This follows directly from the property of PPP. Given the number of points in a homogeneous PPP within a Borel set $B$, the points are independently and uniformly at random distributed in $B$.	
\end{proof}
\begin{remark}
	Since the homogeneous PPP is stationary, the probability density function of the range measurements $f_R(r)$ are identical regardless of which landmark they are obtained from.
\end{remark}

Using Lemma~\ref{lem2}, we then provide the probability that the range measurement $R$ is within distance $\epsilon$ to $r$ in the following lemma.

\begin{lemma}
	\label{lem3}
	Conditioned on the number of visible landmarks, the probability that the distance between the range measurement $R$ and a given $r$ is smaller than $\epsilon$ is given below.
	
	When $0 \leq \epsilon < \frac{d_v}{2}$,
	\begin{align}
		\label{eq:lem3-c1}
		\begin{aligned}
			\nbbP\!&\left[ \left|R - r\right| \leq \epsilon \mid N_\nbx = k\right] = \begin{cases}
				\dfrac{\left(r + \epsilon\right)^2}{d_v^2}, & 0 \leq r < \epsilon, \\
				\dfrac{4\epsilon r}{d_v^2}, & \epsilon \leq r < d_v - \epsilon,\\
				1 - \dfrac{\left(r - \epsilon\right)^2}{d_v^2}, & d_v - \epsilon \leq r \leq d_v,
			\end{cases}
		\end{aligned}
	\end{align}
	where $N_\nbx = N(B_{\nbx}) = \#\{\Phi \cap \nbb\left(\nbx, d_v\right)\} = k > 0$ is the number of visible landmarks at an arbitrary location $\nbx$.
	
	When $\frac{d_v}{2} \leq \epsilon < d_v$,
	\begin{align}
		\label{eq:lem3-c2}
		\begin{aligned}
			\nbbP\!&\left[ \left|R - r\right| \leq \epsilon \mid N_\nbx = k\right] = \begin{cases}
				\dfrac{\left(r + \epsilon\right)^2}{d_v^2}, & 0 \leq r < d_v - \epsilon, \\
				1, & d_v - \epsilon \leq r < \epsilon,\\
				1 - \dfrac{\left(r - \epsilon\right)^2}{d_v^2}, & \epsilon \leq r \leq d_v,
			\end{cases}
		\end{aligned}	
	\end{align}
	
	When $d_v < \epsilon$,
	\begin{align}
		\label{eq:lem3-c3}
		\nbbP\!\left[ \left|R - r\right| \leq \epsilon \mid N_\nbx = k\right]  = 1.
	\end{align}
\end{lemma}

\begin{proof}
	By definition and Lemma~\ref{lem2}, we have 
	\begin{align}
		\begin{aligned}
			\nbbP\!&\left[ \left|R - r\right| \leq \epsilon \mid N_\nbx = k\right] = \int_{r - \epsilon}^{r + \epsilon} f_R(d) \, \nrmd d = \int_{\max \{r - \epsilon, 0\}}^{\min \{r + \epsilon, d_v\}} \frac{2d}{d_v^2} \, \nrmd d,
		\end{aligned}
	\end{align}
	where the last equation is by that $f_R(d)$ only has non-zero values when $d \in [0, d_v]$. The result in Lemma~\ref{lem3} is derived by considering different values of $\epsilon$ and $r_i$. This completes the proof.
\end{proof}

Further, we derive conditional localizability probability $P_{\rm C, Loc}$, which is the probability that the distance between $\nbR_{\nbzero}$ and a given range vector $\nbr_{\nbx} = \left[r_1, \dots, r_k\right]$ is greater than $\epsilon$. The result is presented in the following lemma.
\begin{lemma}
	\label{lem4}
	Given a {\em range vector} $\nbr_{\nbx} \in \nbbR^k$, the conditional localizability probability is
	\begin{align}
		\begin{aligned}
			P_{\rm C, Loc} &= 1 - \nbbP\!\left[\Delta_{\rm p}\!\left(\nbR_{\nbzero}, \nbR_{\nbx}\right)\leq \epsilon \mid \nbR_{\nbx} = \nbr_{\nbx}, N_\nbx = k\right]\\
			&= 1 -  \frac{m^k}{k !} e^{-m}\!\left\{\prod_{i = 1}^{k} \nbbP\!\left[\left|R_i - r_i\right| \leq \epsilon \mid N_\nbzero = k\right]\right\},
		\end{aligned}
	\end{align}
    where $\nbR_{\nbzero} = \left[R_1, \dots, R_k\right]$ is the random range vector obtained at the origin.
\end{lemma}

\begin{proof}
	To calculate $P_{\rm C, Loc}$, we consider the probability of the complementary event, i.e., 
    the obtained range vector $\nbR_{\nbzero}$ at the origin is within $\epsilon$ distance to $\nbr_\nbx$. The probability is 
	\begin{align}
		&\nbbP\!\left[\Delta_{\rm p}\!\left(\nbR_{\nbzero}, \nbR_\nbx\right)\leq \epsilon \mid \nbR_\nbx = \nbr_\nbx, N_\nbx = k\right] \notag\\
		&\overset{(a)}{=} \nbbP\!\left[\Delta_{\rm p}\!\left(\nbR_{\nbzero}, \nbr_\nbx\right)\leq \epsilon, \dim(\nbR_\nbzero) = k \right]\\
		&\overset{(b)}{=}  \nbbP\!\left[\Delta_{\rm p}\!\left(\nbR_{\nbzero}, \nbr_\nbx\right)\leq \epsilon, N(B_{\nbzero}) = k \right] \\
		&\overset{(c)}{=}\nbbP\!\left[\max_{i \in \{1,\dots,N\}}\left\{ \left|R_{i} - r_i\right| \right\} \leq \epsilon \mid N = k \right] \cdot \nbbP\!\left[ N_\nbzero = k  \right]\\
		&\overset{(d)}{=}\nbbP\!\left[\left|R_1 - r_1\right| \leq \epsilon, \dots, \left|R_k - r_k\right| \leq \epsilon \mid N_\nbzero = k\right] \cdot \nbbP\!\left[ N_\nbzero = k\right]\\
		&\overset{(e)}{=} \left\{\prod_{i = 1}^{k} \nbbP\!\left[\left|R_i - r_i\right| \leq \epsilon \mid N_\nbzero = k\right]\right\} \cdot \nbbP\!\left[ N_\nbzero = k \right], \label{eq1:lem3}
	\end{align}
	where $(a)$ follows from the fact that $ \dim(\nbR_\nbzero) = \dim(\nbr_\nbx)$ is the prerequisite of $\Delta_{\rm p}\!\left(\nbR_{\nbzero}, \nbr_\nbx\right)\leq \epsilon$; $(b)$ follows from the fact that the dimension of the range vector is equal to the number of visible landmarks; $(c)$ follows from $\Delta_{\rm p}\!\left(\nbR_{\nbzero}, \nbr_{\nbx}\right) =  \max_{i \in \{1,\dots,N\}} \left|R_{i} - r_i\right| $. Additionally, we employ the definition of the infinity norm to characterize the maximum distance between the elements of $\nbR_{\nbzero}$ and $\nbr_{\nbx}$; in $(d)$ we utilize the $i$-th element of$\nbr_{\nbx}$ to constrain $R_i$; and in $(e)$ the motion-invariance property of the PPP is used, meaning that range measurements from all directions are independent and identical to each other.
	
	The second term in~\eqref{eq1:lem3} is the probability that there are $k$ points in $\Phi$ lying in $B_{\nbzero}$, given as
	\begin{equation}
		\label{eq:pn}
		\nbbP\!\left[ N_\nbzero = k \right] = \frac{ m^k}{k !} e^{-m},
	\end{equation}
	where $ m = \Lambda(B_{\nbzero}) = \int_{B_{\nbzero}} \lambda \, \nrmd \nbx =\lambda \pi d_v^2$ ,  $\lambda$ is the intensity of landmarks, $d_v$ is the maximum visibility distance. This completes the proof.
\end{proof}

Lemma~\ref{lem4} provides the closed-form of conditional localizability probability, which indicates how frequently measurements similar to $\nbr_\nbx$ appear at the origin.
To further calculate the localizability probability, we remove the condition on $\nbr_\nbx$ by first deriving the joint probability density function of $\nbR_\nbx$ and the number of visible landmarks in $B_\nbx$. The result is presented in the following lemma.
\begin{lemma}
	\label{lem5}
	The joint probability density function of $\nbR_\nbx$ and $N_\nbx$ is
	\begin{align}
		f_{\nbR_\nbx, N_\nbx}&\!\left(\nbr_\nbx, k\right) = \frac{m^k}{k !} e^{-m} \prod_{i=1}^{k}\! \left\{\frac{2r_i}{d_v^2} \delta\!\left( 0 \leq r_i  \leq d_v \right)\right\}\!.
	\end{align}
\end{lemma}

\begin{proof}
	Because of the motion-invariance property of PPP, the distribution of the number of points in $B_{\nbx}$ is invariant to the location $\nbx$. As a result, the distribution of $N_\nbx$ is identical to $N_\nbzero$, given in~\eqref{eq:pn}. By definition, the joint probability density function is
	\begin{align}
		f_{\nbR_\nbx, N_\nbx}&\!\left(\nbr_\nbx, k\right) = f_{\nbR_\nbx|N_\nbx}\!\left(\nbr_\nbx|k\right) \cdot \nbbP\!\left[ N_\nbx = k \right].
	\end{align}
	Because PPP is motion-invariant, the distributions of the ranges to visible landmarks from all orientations are independent and identical. Hence, we can write
	\begin{align}
		f_{\nbR_\nbx|N_\nbx}\!\left(\nbr_\nbx|k\right) = \prod_{i=1}^{k} f_R\!\left(r_i\right) =\prod_{i=1}^{k} \left\{\frac{2r_i}{d_v^2}\delta\!\left( 0 \leq r_i  \leq d_v \right)\right\}.
	\end{align}
	This completes the proof.
\end{proof}

Using the previous lemmas and the definition in~\eqref{def:c1-ploc}, we present the main result of this section in the following theorem.
\begin{theorem}
	\label{the1}
	The localizability probability $P_{\rm Loc}$ based on the range vector is
	\begin{align}
		\label{eq:the1}
		P_{\rm Loc} = \nbbP\!\left[ \Delta_{\rm p}\!\left(\nbR_\nbzero, \nbR_\nbx\right) > \epsilon\right] =\begin{cases}
			1 - e^{-2m}\cdot I_0\!\left(2m \cdot \sqrt{\frac{8d_v^3\epsilon - 6 d_v^2 \epsilon^2 + \epsilon^4}{3 d_v^4}}\right), & 0 \!\leq\! \epsilon < d_v,\\
			1 - e^{-2m} \cdot I_0\!\left(2m\right), & \epsilon\! \geq\! d_v,
		\end{cases}
	\end{align}
	where $ m =\lambda \pi d_v^2$ is defined in \eqref{eq:pn}.
\end{theorem}
\begin{proof}
	Using the definition in~\eqref{def:c1-ploc} and Lemma~\ref{lem5}, we have 
	\begin{align}
		\nbbP\!\left[ \Delta_{\rm p}\!\left(\nbR_\nbzero, \nbR_\nbx\right) > \epsilon\right] =\sum_{k=0}^{\infty} \Bigg\{ \int_{0}^{\nbr_v} P\!\left[\Delta_{\rm p}\!\left(\nbR_\nbzero, \nbr_\nbx\right) > \epsilon \mid \nbr_\nbx, k\right] \cdot f_{\nbR_\nbx, N_\nbx}\!\left(\nbr_\nbx, k\right)\, \nrmd \nbr_\nbx \Bigg\}.
	\end{align}
	When $0 \leq \epsilon < d_v$, using equations~\eqref{eq:lem3-c1},~\eqref{eq:lem3-c2} in Lemma~\ref{lem3} and Lemma~\ref{lem4}, the result is
	\begin{align}
		\label{thm1:case1}
		P_{\rm Loc} &= 1 - \sum_{k=0}^{\infty}\! \left\{ \left(\frac{8d_v^3\epsilon - 6 d_v^2 \epsilon^2 + \epsilon^4}{3 d_v^4} \right)^k \cdot \frac{m^{2k}}{(k!)^2} e^{- 2 m}\right\}\\
		&= 1 - e^{-2m} \cdot I_0\!\left(2m \cdot \sqrt{\frac{8d_v^3\epsilon - 6 d_v^2 \epsilon^2 + \epsilon^4}{3 d_v^4}}\right),
	\end{align}
	where $I_0(\cdot)$ is the modified Bessel function of the first kind.
	It should be noted that although equations~\eqref{eq:lem3-c1} and~\eqref{eq:lem3-c2} may appear different, the integrals result in the same value.
	
	When $d_v \leq \epsilon$, using equation~\eqref{eq:lem3-c3}, we have
	\begin{align}
		\label{thm1:case2}
		P_{\rm Loc} =1-\sum_{k=0}^{\infty} \! \left\{ \left(\frac{m^{k}}{k!} e^{- m}\right)^2\right\} = 1-e^{-2m}  \cdot I_0\!\left(2m\right).
	\end{align}
	This completes the proof.
\end{proof}
\begin{cor}
\label{cor1}
When the noise level is significant, meaning $\epsilon \geq d_v$, the localizability probability cannot be improved by including the range vector. 
This can be inferred by comparing  equation~\eqref{thm1:case2} and Lemma~\ref{lem1}, indicating that the range vector does not improve localizability in this regime.
\end{cor}
%%%%%%%%%%%%%%%%%%%%%%%%%%%%%%%%%%%%%%%%%%%%%%%%%%
\begin{remark}
To obtain further insights, the modified Bessel function of the first kind can be bounded within exponential functions, as shown in~\cite{yang2016approximating}, given by
\begin{align}
\frac{e^x}{1 + 2x} < I_0(x) < \frac{e^x}{\sqrt{1 + 2x}}, \quad x > 0.
\end{align}
\end{remark}

Theorem~\ref{the1} provides the analytical expression of localizability probability. A natural question arises about how the localizability probability performs as the landmark intensity $\lambda$ approaches infinity. This question is answered in the following proposition.

\begin{prop}
	\label{prop1}
	As the landmark intensity $\lambda$ tends to infinity, the localizability probability $P_{\rm Loc}$ approaches one, and the non-localizability probability 
    \begin{equation*}
        P_{\rm N-Loc} = 1 - P_{\rm Loc} = O\!\left(e^{-2(1 - \alpha)\lambda \pi d_v^2} \left(4 \alpha \lambda \pi^2  d_v^2\right)^{-\frac{1}{2}} \right),
    \end{equation*}
 approaches zero, regardless of the value of $\epsilon$.
\end{prop}

\begin{proof}
    We use the result provided in~\cite{olver2010nist} and write the asymptotic expansions of the modified Bessel functions of the first kind 
    As $z \rightarrow \infty$ with some fixed $v$, we have
	\begin{align}
		\label{eq:asym}
		I_v(z) = \frac{e^z}{\left(2\pi z\right)^{\frac{1}{2}}} \sum_{k=0}^{\infty} (-1)^k \left(\frac{a_k(v)}{z^k}\right),
	\end{align}
	where
	\begin{align}
		a_k(v) =\begin{cases}
			1, & k = 0,\\
			\frac{(4v^2 - 1^2)(4v^2 - 3^2)\dots(4v^2 - (2k-1)^2)}{k! 8k}, & k \neq 0.
		\end{cases}
	\end{align}
    The non-localizability probability under the limit $\lambda \rightarrow \infty$ can be rewritten using~\eqref{eq:asym} with $v = 0$ and $z = 2 \alpha m$. It becomes
	\begin{align}
		\lim_{\lambda \rightarrow \infty} P_{\rm N-Loc} &= \lim_{\lambda \rightarrow \infty} 1 -  P_{\rm Loc}\\
		&= \lim_{\lambda \rightarrow \infty} 1 - \nbbP\!\left[ \Delta_{\rm p}\!\left(\nbR_\nbzero, \nbR_\nbx\right) > \epsilon\right]  \notag\\
		&= \lim_{m \rightarrow \infty} e^{-2m} \cdot \frac{e^{2\alpha m}}{\left(4\pi \alpha m \right)^{\frac{1}{2}}} \left(1 + O\left(2\alpha m^{-\frac{1}{2}}\right)\right) \\
		&= \lim_{m \rightarrow \infty}  e^{-2(1 - \alpha)m} \left(4\pi \alpha m\right)^{-\frac{1}{2}} \left(1 + O\left(m^{-\frac{1}{2}}\right)\right)\\
		&= 0, 
	\end{align}
	where $m = \lambda \pi d_v^2 $ and 
	\begin{align}
		\alpha
		&= \begin{cases}
			\sqrt{\frac{8d_v^3\epsilon - 6 d_v^2 \epsilon^2 + \epsilon^4}{3 d_v^4}}, & 0 \leq \epsilon < d_v,\\
			1, & \epsilon \geq d_v,
		\end{cases}
	\end{align}
	which is bounded by $0 \leq \alpha \leq 1$. Therefore, we have $ P_{\rm N-Loc} \in O\!\left(e^{-2(1 - \alpha)m} \left(4\pi \alpha m\right)^{-\frac{1}{2}} \right) $
	This completes the proof.
\end{proof}

\section{Case 2: Set of Relative Locations}\label{sec:case2}
This section considers the scenario where the target is an aircraft equipped with a camera that captures aerial images, such as in the case of unmanned aerial vehicles.
The target aims to determine its location on the ground using the captured aerial image.
The locations of visible landmarks relative to the target's ground location can be estimated using the aerial image.
Since there is no specific ordering for these relative positions, we use a set to represent them, defined as:
\begin{equation}
	\label{def:set_relative_loc}
	\ncalX_{\nbx} = g(\nbx) = \left\{\tilde{\nbx}_i \mid \tilde{\nbx}_i = \nbx_i - \nbx, \nbx_i \in \Phi \cap B_{\nbx}\right\},
\end{equation}
where $\ncalX_{\nbx}$ is equivalent to the landmark pattern at location $\nbx$.
\begin{figure}
	\centering
	\includegraphics[width=0.3\textwidth]{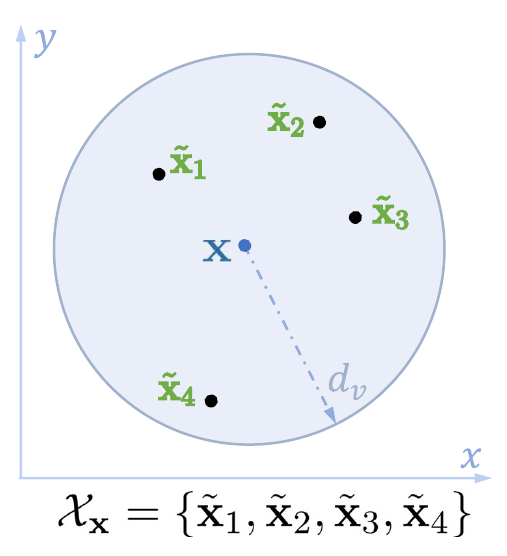}
	\caption{An illustration of the set of relative locations measured from location $\nbx$.}
	\label{fig:c2-model}
\end{figure}
\figref{fig:c2-model} illustrates the set of relative locations measured from location $\nbx$.
Since there is no specific ordering, we consider all possible permutations of the elements in the set and represent them as vectors.
The resulting set of vectors is denoted as $S(\ncalX_{\nbx})$, where $S(\cdot)$ represents all possible permutations. For instance, all possible permutations of $\{1,2,3\}$ using the operator $S(\cdot)$ is
\begin{equation}
	S(\{1,2,3\}) = \{[1,2,3], [1,3,2], [2,1,3], [2,3,1], [3,1,2], [3,2,1]\},
\end{equation}
which contains $3!$ permutations of three elements. 
We consider $S(\ncalX_{\nbx})$ as a function of location $\nbx$, denoted by $f(\nbx)$.
The mapping $f$ is formally defined as
\begin{equation}
	f: \nbbR^2 \mapsto \nbbR^{N_\nbx \times 2}, \nbx \mapsto f(\nbx) = S(\ncalX_{\nbx}) = \{\nbv_1, \dots, \nbv_i, \dots, \nbv_{N_\nbx!}\},
\end{equation}
where $N_\nbx$ represents the number of visible landmarks at location $\nbx$ and $\nbv_i = [\tilde{\nbx}_1, \dots, \tilde{\nbx}_{N_\nbx}]$ represents the permutation of the relative locations of landmarks. Since landmarks are positioned on a two-dimensional plane $\nbbR^2$, $\nbv_i$ is a matrix in the space $\nbbR^{N_\nbx \times 2}$. 
Next, we formalize the notion of distance between two matrices:
\begin{equation}
	\label{def:C2-vec-dist}
	\Delta_{\rm v}\!\left(\nbv_i, \nbv_j\right) = 
	\begin{cases}
		\left\| \nbv_i - \nbv_j \right\|_{2, \infty}, & \dim(\nbv_i) = \dim(\nbv_j) \\
		\infty, & \dim(\nbv_i) \neq \dim(\nbv_j)
	\end{cases},
\end{equation}
where $\left\| \nbv \right\|_{2, \infty} = \max_{1 \leq i \leq N_\nbx} \left\|\nbv_{i:} \right\|_2$ is the norm of matrix induced by the $l_2$ and infinity vector norm. 
Based on the matrix distance $\Delta_{\rm v}$, we further define the distance between two sets of permutations
\begin{align}
	\label{def:C2-set-dist}
	\Delta_{\rm s}\!\left(f(\nbzero), f(\nbx)\right) = \min_{\nbv_{i} \in f(\nbzero), \nbv_j \in f(\nbx)} \Delta_{\rm v}\!\left(\nbv_i, \nbv_j\right),
\end{align}
which is the minimal matrix distance among matrices in both sets. 

The measurement error, denoted by $\nbn \in \nbbR^{N_\nbx \times 2}$, is introduced by adding it to permutation matrices $\nbv$. 
We can define the set of possible noisy measurements observed at location $\nbx$, given by
\begin{equation}
	F_{\epsilon}(\nbx) = \left\{\nbv + \nbn \mid \nbv \in f(\nbx), \nbn \in \nbbR^{N_\nbx \times 2}, \| \nbn \|_{2, \infty} \leq \epsilon/2\right\}.
\end{equation}
Since $\| \nbn \|_{2, \infty} \leq \epsilon/2$, it becomes evident that the $l_2$ norm of each column of $\nbn$, denoted as $\left\| \nbn_{i,:} \right\|_2$, is bounded $\epsilon /2$. 
Consequently, the landmark locations affected by this noise deviate by a maximum Euclidean distance of $\epsilon /2 $ from their original positions.  $F_{\epsilon}(\nbx)$ can be equivalently represented using the matrix distance in~\eqref{def:C2-vec-dist}
\begin{equation}
	F_{\epsilon}(\nbx) = \left\{\tilde{\nbv} \mid \Delta_{\rm v}\!\left(\tilde{\nbv},\nbv \right) \leq \epsilon/2, \nbv \in f(\nbx)\right\},
\end{equation}
where $\tilde{\nbv}$ are the permutations of noisy relative locations of the visible landmarks at location $\nbx$.

We can now express the non-localizability probability based on the relative locations
\begin{align}
	\label{eqdef:C2-P-nloc}
	P_{\rm N-Loc} &= \nbbP\!\left[F_{\epsilon}(\nbx) = \varnothing, F_{\epsilon}(\nbzero) = \varnothing \right] + \nbbP\!\left[F_{\epsilon}(\nbx) \cap F_{\epsilon}(\nbzero) \neq \varnothing \right]\\
	&= \nbbP\!\left[N\!\left( B_{\nbx}\right)  = 0, N\!\left( B_{\nbzero}\right)  = 0 \right] + \nbbP\!\left[\Delta_{\rm s}\!\left(F(\nbx), F(\nbzero)\right) \leq \epsilon \right], \label{eqdef:C2-P-nloc-p1}
\end{align}
where $\nbbP\!\left[N\!\left( B_{\nbx}\right)  = 0, N\!\left( B_{\nbzero}\right)  = 0 \right]$ is an additional term that the visible landmarks do not exist. 
We consider $F(\nbzero)$ and $F(\nbx)$ as two random noise-free measurements obtained at the origin and $\nbx$, respectively, since the landmark locations form a PPP and are not deterministic.
The first term in~\eqref{eqdef:C2-P-nloc-p1} represents the probability that no landmarks are visible at both locations, and the second component represents the probability that the noisy measurements are overlapped.
%Mathematically, the derivation of the first component is simplified by acknowledging the stationary nature of the PPP and the assumption that $\nbX$ isn't closely situated to $\nbX_0$. 
%Given these conditions, it is logical to posit that the regions $B_{\nbX_0}$ and $B_{\nbX}$ do not intersect, represented as $B_{\nbX_0} \cap B_{\nbX} = \varnothing$. 
As discussed in Remark~\ref{rem:localizability}, the probabilities related to the number of visible landmarks are almost surely independent. We have
\begin{align}
	\label{eqref:C2-first-term}
	\nbbP&\! \left[N\!\left( B_{\nbx}\right)  = 0, N\!\left( B_{\nbzero}\right) = 0\right] = \nbbP\!\left[N\!\left( B_{\nbx}\right) = 0 \right] \cdot \nbbP\!\left[ N\!\left( B_{\nbzero}\right)  = 0 \right] = e^{- 2m}, 
\end{align}
where $m = \lambda \pi d_v^2$.
We calculate the second component in~\eqref{eqdef:C2-P-nloc-p1} by fixing $F(\nbx) = f(\nbx)$. The result is presented in the following lemma.
\begin{lemma}
	\label{lem6}
	Given the measurement $F(\nbx) = f(\nbx)$, the probability of $\Delta_{\rm s}\!\left(f(\nbx), F(\nbzero)\right) \leq \epsilon$ can be formulated as:
	\begin{align}
		\nbbP\!&\left[ \Delta_{\rm s}(f(\nbx), F(\nbzero)) \leq \epsilon \right] = \frac{m^k}{k!} e^{-m} \cdot \nbbE_{\Phi \cap B_\nbzero \mid N_\nbzero}\! \left[ \lor_{\nbV_0 \in F(\nbzero)} \left\{\prod_{\tilde{\nbX}_i \in \nbV_0} \nb1\!\left( \tilde{\nbX}_i \in A_i \right)\right\} \right],
	\end{align}
where $m = \lambda \pi d_v^2$, $A_i = \nbb(\tilde{\nbx}_i, \epsilon)$ is a ball centered at $\tilde{\nbx}_i$ with radius $\epsilon$, $\tilde{\nbx}_i \in f(\nbx)$ is the relative location of the visible landmark observed at $\nbx$, the symbol $\lor$ represents the logical OR operation.
\end{lemma}

\begin{proof}
	The probability of the event can be written as the expectation of an indicator function, given as:
	\begin{align}
		\nbbP\!\left[ \Delta_{\rm s}(f(\nbx), F(\nbzero)) \leq \epsilon \right]
		= \nbbE_{\Phi \cap B_\nbzero} \! \left[\nb1\!\left(\Delta_{\rm s}(f(\nbx), F(\nbzero)) \leq \epsilon\right)\right].
	\end{align}
	The inequality $\Delta_{\rm s}(f(\nbx), F(\nbzero)) \leq \epsilon$ means that there exists a matrix $\nbV_0 \in F(\nbzero)$ such that $\Delta_{\rm v}\!\left(\nbv, \nbV_0\right) \leq \epsilon$ holds for some $\nbv \in f(\nbx)$. Thus, the corresponding indicator function is
	\begin{align}
		\nb1\!\left(\Delta_{\rm s}(f(\nbx), F(\nbzero)) \leq \epsilon\right) = \lor_{\nbV_0 \in F(\nbzero)} \nb1\!\left(\Delta_{\rm v}(\nbv, \nbV_0) \leq \epsilon\right).
	\end{align}
    The above indicator function equals $1$, if and only if 
    one of the distances $\Delta_{\rm v}(\nbv, \nbV_0)$ is less than or equal to $\epsilon$.
    Based on the definition of the matrix distance in~\eqref{def:C2-vec-dist}, $\Delta_{\rm v}(\nbv, \nbV_0) \leq \epsilon$ indicates that the $l_2$ norm of columns in $\Delta\nbV = \nbV_0-\nbv$, denoted by $\left\| [\Delta\nbV]_{i,:} \right\|_2$, are bounded $\epsilon$. We have
	\begin{align}
		\nb1\!\left(\Delta_{\rm v}(\nbv, \nbV_0) \leq \epsilon\right) &= \nb1\!\left({\rm dim}(\nbV_0) = {\rm dim}(\nbv) \right) \cdot \prod_{\tilde{\nbX}_i \in \nbV_0} \nb1\!\left(\left\| [\Delta\nbV]_{i,:} \right\|_2 \leq \epsilon\right) \\
		&= \nb1\!\left(N(B_\nbzero) = k \right) \cdot \prod_{\tilde{\nbX}_i \in \nbV_0} \nb1\!\left( \left\| \tilde{\nbX}_i - \tilde{\nbx}_i \right\|_2 \leq \epsilon\right), \label{eqref:c3-1}
	\end{align}
	where $\tilde{\nbX}_i$ and $\tilde{\nbx}_{i}$ are the $i$-th columns of metrics $\nbV_0$ and $\nbv$, respectively, $N_\nbzero$ is the number of visible landmarks at the origin. 
    To simplify, we define the ball $A_i = \nbb(\tilde{\nbx}_i, \epsilon)$, where locations within $A_i$ are within $\epsilon$ distance to $\tilde{\nbx}_i$. The equation \eqref{eqref:c3-1} can be simplified
	\begin{equation}
        \label{eqref:c2-region}
		\nb1\!\left(\Delta_{\rm v}(\nbv, \nbV_0) \leq \epsilon\right) = \nb1\!\left(N(B_\nbzero) = k \right) \cdot \prod_{\tilde{\nbX}_i \in \nbV_0} \nb1\!\left( \tilde{\nbX}_i \in A_i \right).
	\end{equation}
	Now, we give the probability of $\Delta_{\rm s}\!\left(f(\nbx), F(\nbzero)\right) \leq \epsilon$
	\begin{align}
		\nbbP\!\left[ \Delta_{\rm s}\!\left(f(\nbx), F(\nbzero)\right) \leq \epsilon \right]
		&= \nbbE_{\Phi \cap B_\nbzero}\! \left[ \nb1\!\left(N(B_\nbzero) = k \right) \cdot \lor_{\nbV_0 \in F(\nbzero)} \left\{\prod_{\tilde{\nbX}_i \in \nbV_0} \nb1\!\left( \tilde{\nbX}_i \in A_i \right)\right\} \right] \\
		&= \nbbP\! \left[ N_\nbzero = k \right] \cdot \nbbE_{\Phi \cap B_\nbzero \mid N_\nbzero}\! \left[ \lor_{\nbV_0 \in F(\nbzero)} \left\{\prod_{\tilde{\nbX}_{0,i} \in \nbV_0} \nb1\!\left( \tilde{\nbX}_i \in A_i \right)\right\} \right] \\
		&= \frac{m^k}{k!} e^{-m} \cdot \nbbE_{\Phi \cap B_\nbzero \mid N_\nbzero}\! \left[ \lor_{\nbV_0 \in F(\nbzero)} \left\{\prod_{\tilde{\nbX}_i \in \nbV_0} \nb1\!\left( \tilde{\nbX}_i \in A_i \right)\right\} \right],
	\end{align}
	where $m = \lambda \pi d_v^2$. This completes the proof.
\end{proof}

The regions $A_i$ in Lemma~\ref{lem6} may overlap, leading to results that are not easily tractable.
Thus, we explore its property by deriving a lower bound in the following lemma.
\begin{lemma}
	\label{lem7}
	Given the measurement $f(\nbx)$, the lower bound of the probability of $\Delta_{\rm s}(f(\nbx), F(\nbzero)) \leq \epsilon$ is
	\begin{equation}
		\nbbP\!\left[ \Delta_{\rm s}(f(\nbx), F(\nbzero)) \leq \epsilon \right] \geq \frac{m^k}{k!} e^{-m} \cdot \prod_{i=1}^k \left\{\frac{|A_i|}{\pi d_v^2}\right\}.
	\end{equation}
\end{lemma}

\begin{proof}
	The result of logical OR operation is lower bounded by 
	\begin{align}
		\lor_{\nbV_0 \in F(\nbzero)} \left\{\prod_{\tilde{\nbX}_i \in \nbV_0} \nb1\!\left( \tilde{\nbX}_i \in A_i \right)\right\} \geq \prod_{\tilde{\nbX}_i \in \nbV_0^*} \nb1\!\left( \tilde{\nbX}_i \in A_i \right),
	\end{align}
	where $\nbV_0^*$ is an arbitrary elements in $F(\nbzero)$. With the above inequality, we represent the expectation in Lemma~\ref{lem6} as
	\begin{align}
		\nbbE_{\Phi \cap B_\nbzero \mid N_\nbzero}&\! \left[ \lor_{\nbV_0 \in F(\nbzero)} \left\{\prod_{\tilde{\nbX}_i \in \nbV_0} \nb1\!\left( \tilde{\nbX}_i \in A_i \right)\right\} \right] \\
        &\geq \nbbE_{\Phi \cap B_\nbzero \mid N_\nbzero}\! \left[ \prod_{\tilde{\nbX}_i \in \nbV_0^*} \nb1\!\left( \tilde{\nbX}_i \in A_i \right) \right] \\
		&= \nbbP\! \left[ \cap_{ \tilde{\nbX}_i \in \nbV_0^*} \left\{ \tilde{\nbX}_i \in A_i \right\} \mid N_\nbzero = k \right].
	\end{align}
	Now, given the number of landmarks $N_\nbzero$, the locations of these landmarks are independently and uniformly distributed in the visibility region $B_\nbzero$. Under this condition, the probability above can be written as
	\begin{align}
		\nbbP&\! \left[ \cap_{ \tilde{\nbX}_i \in \nbV_0^*} \left\{ \tilde{\nbX}_i \in A_i \right\} \mid N_\nbzero = k \right] \\
        &= \prod_{\tilde{\nbX}_i \in \nbV_0^*} \nbbP\! \left[ \tilde{\nbX}_i \in A_i \mid N_\nbzero = k \right] =  \prod_{1 \leq i \leq k} \frac{|A_i|}{\pi d_v^2},
	\end{align}
	where $|A_i|$ represents the Lebesgue measure of the ball $A_i$. This completes the proof.
\end{proof}

Additionally, we can derive an upper bound of the probability in Lemma~\ref{lem6}. The result is as follows:
\begin{lemma}
\label{lem8}
Given the measurement $f(\nbx)$, the upper bound of the probability of $\Delta_{\rm s}(f(\nbx), F(\nbzero)) \leq \epsilon$ is
\begin{equation}
	\nbbP\! \left[ \Delta_{\rm s}(f(\nbx), F(\nbzero)) \leq \epsilon \right] \leq m^k e^{-m} \cdot \prod_{i=1}^{k} \left\{\frac{|A_i|}{\pi d_v^2}\right\}.
\end{equation}
\end{lemma}
\begin{proof}
	The result of logical OR operation is upper bounded by directly summing of the indicators, given as 
	\begin{align}
		\lor_{\nbV_0 \in F(\nbzero)} \left\{\prod_{\tilde{\nbX}_i \in \nbV_0} \nb1\!\left( \tilde{\nbX}_i \in A_i \right)\right\} \leq \sum_{\nbV_0 \in F(\nbzero)} \left\{\prod_{\tilde{\nbX}_i \in \nbV_0} \nb1\!\left( \tilde{\nbX}_i \in A_i \right),\right\}.
	\end{align}
    With the techniques in Lemma~\ref{lem7}, it is straightforward to complete this proof.
\end{proof}

Using results from Lemma~\ref{lem7} and Lemma~\ref{lem8}, we present the bounds on the localizability probability.

\begin{theorem}
	\label{the2}
		The localizability probability $P_{\rm Loc}$ based on the relative locations of landmarks is upper and lower bounded by 
	\begin{align}
		1 - \exp\!\left(-2m\right) I_0\!\left(\frac{2m\epsilon}{d_v}\right)\geq P_{\rm Loc} = 1 - P_{\rm N-Loc} \geq 1 - \exp\!\left(\frac{m^2\epsilon^2}{d_v^2} - 2m\right),
	\end{align}
	where $m = \lambda \pi d_v^2$ and $I_0(\cdot)$ is the modified Bessel function of the first kind. 
\end{theorem}

\begin{proof}
	We can write the localizability probability by taking the expectation of the conditional localizability probability over $\Phi \cap B_\nbx$, given as 
	\begin{align}
		\nbbP\! \left[ \Delta_{\rm s}(F(\nbx), F(\nbzero)) \leq \epsilon  \right] &= \nbbE_{\Phi \cap B_\nbx}\!\left[\nbbP\! \left[ \Delta_{\rm s}(F(\nbx), F(\nbzero)) \leq \epsilon \mid F(\nbx) = f(\nbx) \right]\right].
	\end{align}
	Using the result from Lemma~\ref{lem7}, we have
	\begin{align}
		\nbbE_{\Phi \cap B_\nbx}&\!\left[\nbbP\! \left[ \Delta_{\rm s}(F(\nbx), F(\nbzero)) \leq \epsilon \mid F(\nbx) = f(\nbx)\right]\right] \\
        &\overset{(a)}{\geq} \nbbE_{\Phi \cap B_\nbx}\left[\frac{m^{N_\nbx}}{N_\nbx!} e^{-m} \cdot \left\{\frac{\epsilon^2}{d_v^2}\right\}^{N_\nbx}\right] \\
		&\overset{(b)}{=} \sum_{k = 1}^{\infty} \left(\frac{m^k}{k!} e^{-m}\right)^2 \cdot \left\{\frac{\epsilon^2}{d_v^2}\right\}^k \\
		&= e^{-2m} \left(I_0\!\left(\frac{2m\epsilon}{d_v}\right) - 1\right),
	\end{align}
	where (a) follows by the fact that $|A_i| = |b(\tilde{\nbx}_i, \epsilon)| = \pi \epsilon^2$ and $N_\nbx$ is the number of landmarks in $B_\nbx$; (b) follows by $N(B_\nbx) = k$ and we remove the condition on the locations of the points in $B_\nbx$.
	Now, we use the definition of the non-localizability probability in~\eqref{eqdef:C2-P-nloc} and have
	\begin{align}
		P_{\rm N-Loc} &= \nbbP\!\left[N\!\left( B_{\nbx}\right)  = 0, N\!\left( B_{\nbzero}\right)  = 0 \right] + \nbbP\!\left[\Delta_{\rm s}\!\left(F(\nbx), F(\nbzero)\right) \leq \epsilon\right]\\
		&\geq e^{-2m} \left(I_0\!\left(\frac{2m\epsilon}{d_v}\right) - 1\right) + e^{-2m} = e^{-2m} I_0\!\left(\frac{2m\epsilon}{d_v}\right). \label{eq:the2-lb}
	\end{align}
	The left side of the inequality is proved. 
    The right side of the inequality can be proved using Lemma~\ref{lem8} with the same technique. This completes the proof.
\end{proof}
\begin{comment}
    \begin{remark}
	Suppose the relative location measurements are in a vector form, similar to the approach we used in section~\ref{sec:case1}. In that case, we can directly derive the corresponding localizability probability, which will have the same expression of the lower bound in~\eqref{eq:the2-lb}. This result is reasonable since the vector form contains ordering information among these measurements.
\end{remark}
\end{comment}

\section{Case 3: Set of Ranges} \label{sec:case3}
In the previous section, we discussed the scenario involving the use of the relative locations of visible landmarks.
However, the relative locations may not be available in many situations, for example, when the target is on the ground and takes images of its surroundings to estimate the ranges to visible landmarks. In such cases, the orientations are unknown and the relative locations cannot be determined.
Without ordering, we represent these range measurements as a set
\begin{equation}
	\ncalD_{\nbx} = \left\{ r_i \mid r_i = \|\nbx_i - \nbx\|_2, \nbx_i \in \Phi \cap B_{\nbx} \right\}.
\end{equation}
We illustrate the set of range measurements at location $\nbx$ in~\figref{fig:c3-model}.
\begin{figure}
	\centering
	\includegraphics[width=0.3\textwidth]{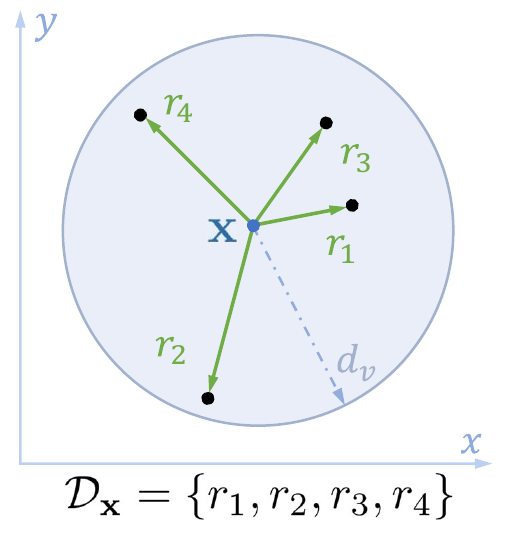}
	\caption{An illustration of the set of range measurements at location $\nbx$.}
	\label{fig:c3-model}
\end{figure}
Mathematically, this representation of range measurements shares the same structure as defined in~\eqref{def:set_relative_loc}. 
Similarly, we formulate these range measurements into the set of vectors, denoted as $S(\ncalD_{\nbx})$. 
We treat $S(\ncalD_{\nbx})$ as the function of location $\nbx$, defined as
\begin{equation}
	f: \nbbR^2 \mapsto \nbbR^{N_\nbx}, \nbx \mapsto f(\nbx) = S(\ncalD_{\nbx}) = \left\{\nbr_1, \cdots,\nbr_{N_\nbx!}\right\},
\end{equation}
where $N_\nbx$ is the number of visible landmarks at location $\nbx$ and $\nbr_i = \left[r_1, \cdots, r_{N_\nbx}\right]$ is
a vector representing the permutation of range measurements. 
Similar to the distance defined in~\eqref{def:C2-vec-dist}, we define the distance between permutations of the range measurement as
\begin{equation}
	\label{def:C3-vec-dist}
	\Delta_{\rm v}\!\left(\nbr_i, \nbr_j\right) = 
	\begin{cases}
		\left\| \nbr_i - \nbr_j \right\|_{\infty}, & \dim(\nbr_i) = \dim(\nbr_j) \\
		\infty, & \dim(\nbr_i) \neq \dim(\nbr_j)
	\end{cases}.
\end{equation}
Further, the distance between two sets of permutations is defined as
\begin{align}
	\label{def:C3-set-dist}
	\Delta_{\rm s}\!\left(f(\nbzero), f(\nbx)\right) = \min_{\nbr_i \in f(\nbzero), \nbr_j \in f(\nbx)} \Delta_{\rm v}\!\left(\nbr_i, \nbr_j\right).
\end{align}
The measurement error is introduced by adding a noise $\nbn \in \nbbR^{N_\nbx}$ to the range vector $\nbr$. 
The set of possible noisy range vectors observed at location $\nbx$ is
\begin{equation}
	F_{\epsilon}(\nbx) = \left\{\nbr + \nbn \mid \nbr \in f(\nbx), \nbn \in \nbbR^{N_\nbx}, \| \nbn \|_{\infty} \leq \epsilon/2\right\}
\end{equation}
Since $\| \nbn \|_{\infty} \leq \epsilon/2$, the distance between the noisy range vector $\tilde{\nbr} = \nbr + \nbn$ and its noise-free version $\nbr$ is bounded $\epsilon /2$. 
Thus, $F_{\epsilon}(\nbx)$ can be equivalently represented with the vector distance, given as
\begin{equation}
	F_{\epsilon}(\nbx) = \left\{\tilde{\nbr} \mid \Delta_{\rm v}\!\left(\tilde{\nbr},\nbr \right) \leq \epsilon/2, \nbr \in f(\nbx)\right\}.
\end{equation}
We now define the localization probability based on the set of ranges as
\begin{align}
	\label{eqdef:C3-P-nloc}
	P_{\rm N-Loc} &= \nbbP\!\left[F_{\epsilon}(\nbx) = \varnothing, F_{\epsilon}(\nbzero) = \varnothing \right] + \nbbP\!\left[F_{\epsilon}(\nbx) \cap F_{\epsilon}(\nbzero) \neq \varnothing \right]\\
	&= \nbbP\!\left[N\!\left( B_{\nbx}\right)  = 0, N\!\left( B_{\nbzero}\right)  = 0 \right] + \nbbP\!\left[\Delta_{\rm s}\!\left(F(\nbx), F(\nbzero)\right) \leq \epsilon \right],
\end{align}
where $F(\nbzero)$ and $F(\nbx)$ are two random measurements obtained at the origin and $\nbx$, respectively, since the landmark locations form a PPP and are not deterministic.
The first term of~\eqref{eqdef:C3-P-nloc} is given in~\eqref{eqref:C2-first-term}, and the second component represents the probability that the noisy measurements are overlapped.

We calculate the second component by fixing F(x) = f(x) and present in the following lemma.
\begin{lemma}
	\label{lem9}
	Given the measurement $F(\nbx) = f(\nbx)$, the probability of $\Delta_{\rm s}\!\left(f(\nbx), F(\nbzero)\right) \leq \epsilon$ can be formulated as:
	\begin{align}
		\nbbP\!&\left[ \Delta_{\rm s}\!\left(f(\nbx), F(\nbzero)\right) \leq \epsilon \right] \notag \\
		&= \frac{m^k}{k!} e^{-m} \cdot \nbbE_{\Phi \cap B_\nbzero \mid N_\nbzero}\! \left[ \lor_{\nbR_0 \in F(\nbzero)} \left\{\prod_{R_i = \|\tilde{\nbX}_i\|_2 \in \nbR_0} \nb1\!\left( \tilde{\nbX}_i \in A_i \right)\right\} \right].
	\end{align}
\end{lemma}

\begin{proof}
	Similar to the proof for Lemma~\ref{lem6}, we write the probability as the expectation of the indicator function
	\begin{align}
		\nbbP\!\left[ \Delta_{\rm s}(f(\nbx), F(\nbzero)) \leq \epsilon\right]
		= \nbbE_{\Phi \cap B_\nbzero} \! \left[\nb1\!\left(\Delta_{\rm s}(f(\nbx), F(\nbzero)) \leq \epsilon\right)\right].
	\end{align}
    With the distance defined in~\eqref{def:C3-set-dist}, the event $\left\{\Delta_{\rm s}(f(\nbx), F(\nbzero)) \leq \epsilon\right\}$ indicates that there exists a vector $\nbR_0 \in F(\nbzero)$ such that $\Delta_{\rm v}\!\left(\nbr, \nbR_0\right) \leq \epsilon$ holds for the vector $\nbr \in f(\nbx)$.
    The corresponding indicator function is
	\begin{align}
		\nb1\!\left(\Delta_{\rm s}(f(\nbx), F(\nbzero)) \leq \epsilon\right) = \lor_{\nbR_0 \in F(\nbzero)} \nb1\!\left(\Delta_{\rm v}(\nbr, \nbR_0) \leq \epsilon\right),
	\end{align}
	where the symbol $\lor$ represents the logical OR operation. If one of the distances $\Delta_{\rm v}(\nbr, \nbR_0)$ is less than or equal to $\epsilon$, the logical OR of these indicator functions goes to one.
	Further, by the definition in~\eqref{def:C3-vec-dist},  $\Delta_{\rm v}(\nbr, \nbR_0) \leq \epsilon$ indicates that the difference of ranges in each dimension $| R_i - r_i | \leq \epsilon, 1\leq i \leq N_\nbx$. 
    The corresponding indicator functions are
	\begin{align}
		\nb1\!\left(\Delta_{\rm v}(\nbr, \nbR_0) \leq \epsilon\right) &= \nb1\!\left({\rm dim}(\nbR_0) = {\rm dim}(\nbr) \right) \cdot \prod_{R_i \in \nbR_0} \nb1\!\left(| R_i - r_i | \leq \epsilon\right) \\
		&= \nb1\!\left(N(B_\nbzero) = k \right) \cdot \prod_{R_i = \|\tilde{\nbX}_i\|_2 \in \nbR_0} \nb1\!\left( | \|\tilde{\nbX}_i\|_2 - \|\tilde{\nbx}_i\|_2 | \leq \epsilon\right), \label{eqref:c3-2}
	\end{align}
	where $\tilde{\nbX}_i$ and $\tilde{\nbx}_{i}$ are the relative locations of the visible landmarks, $N(\nbzero) = N_\nbzero$ is the number of visible landmarks at the origin.
    One can observe that the equation~\eqref{eqref:c3-2} has the same form as~\eqref{eqref:c3-1} by replacing $| \|\tilde{\nbX}_i\|_2 - \|\tilde{\nbx}_i\|_2 | \leq \epsilon$ with $\left\| \tilde{\nbX}_i - \tilde{\nbx}_i \right\|_2 \leq \epsilon$, and the correspond regions are the annuli $A_i = \nbb(\nbzero, \|\tilde{\nbx}_i\|_2 + \epsilon) \backslash \nbb(\nbzero, \|\tilde{\nbx}_i\|_2 - \epsilon)$. Now, \eqref{eqref:c3-2} can be simplified
    \begin{equation}
		\nb1\!\left(\Delta_{\rm v}(\nbr, \nbR_0) \leq \epsilon\right) = \nb1\!\left(N(B_\nbzero) = k \right) \cdot \prod_{R_i = \|\tilde{\nbX}_i\|_2 \in \nbR_0} \nb1\!\left( \tilde{\nbX}_i \in A_i \right).
	\end{equation}
    The follow-up proof is the same as the proof of Lemma~\ref{lem7} and the conditional probability 
	\begin{align}
		\nbbP\!\left[ \Delta_{\rm s}(f(\nbx), F(\nbzero)) \leq \epsilon\right] =\frac{m^k}{k!} e^{-m} \cdot \nbbE_{\Phi \cap B_\nbzero \mid N_\nbzero}\! \left[ \lor_{\nbR_0 \in F(\nbzero)} \left\{\prod_{R_i = \|\tilde{\nbX}_i\|_2 \in \nbR_0} \nb1\!\left( \tilde{\nbX}_i \in A_i \right)\right\} \right],
	\end{align}
	where $m = \lambda \pi d_v^2$. This completes the proof.
\end{proof}

Because of the overlaps among $A_i$, the derivation of the exact expression of the conditional localizability probability in Lemma~\ref{lem9} is not tractable. Thus, we provide its lower bound in the following lemma.
\begin{lemma}
	\label{lem10}
	Given the measurement $f(\nbx)$, the lower bound of the probability is
	\begin{equation}
		\nbbP \left[ \Delta_{\rm s}(f(\nbx), F(\nbzero)) \leq \epsilon \right] \geq \frac{m^k}{k!} e^{-m} \cdot \prod_{i = 1}^{k} \left\{\frac{|A_i|}{\pi d_v^2}\right\}.
	\end{equation}
\end{lemma}

\begin{proof}
	The proof strategy is the same as Lemma~\ref{lem8}. The difference is $A_i = \nbb(\nbzero, \|\tilde{\nbX}_i\|_2 + \epsilon) \backslash \nbb(\nbzero, \|\tilde{\nbX}_i\|_2 - \epsilon)$ in this case.
\end{proof}

Additionally, we provide an upper bound of the conditional localizability probability as follows:
\begin{lemma}
	\label{lem11}
	Given the measurements $f(\nbx)$, the upper bound of the probability is
	\begin{equation}
		\nbbP\! \left[ \Delta_{\rm s}(f(\nbx), F(\nbzero)) \leq \epsilon \right] \leq m^k e^{-m} \cdot \prod_{i = 1}^{k} \left\{\frac{|A_i|}{\pi d_v^2}\right\}.
	\end{equation}
\end{lemma}
\begin{proof}
	The approach for proving this lemma is the same method used in Lemma~\ref{lem8}. The only difference is $A_i = \nbb(\nbzero, \|\tilde{\nbX}_i\|_2 + \epsilon) \backslash \nbb(\nbzero, \|\tilde{\nbX}_i\|_2 - \epsilon)$ in this case.
\end{proof}

\begin{lemma} % Given the distribution of Area
	\label{lem12}
	Conditioned on the number of visible landmarks, the expectation of the Lebesgue measure of the region $A_i$ is 
	\begin{equation}
		\nbbE_{R_i | N_\nbx}\left[|A_i|\right] = \frac{\pi \epsilon}{3 d_v^2} \left(8 d_v^3 - 6 d_v^2 \epsilon + \epsilon^3\right), \quad 1 \leq i \leq N_\nbx,
	\end{equation}
	where $N_\nbx$ is the number of visible landmarks at location $\nbx$, $R_i = \|\tilde{\nbx}_i\|$ is the range to the $i$-th visible landmark.
\end{lemma}
\begin{proof}
	When conditioned on the number of visible landmarks, the locations of these landmarks follow a uniform distribution within the visibility region. The Lebesgue measure of the annular is a function of the range $r$, given by
	\begin{equation}
		s(r) = \begin{cases}
			\pi \!\left(r+\epsilon\right)^2, & 0 < r \leq \epsilon\\
			4 \pi r \epsilon, &\epsilon < r \leq d_v - \epsilon \\
			\pi \! \left(d_v^2 - r^2 - \epsilon^2 + 2r\epsilon\right), & d_v - \epsilon < r < d_v.
		\end{cases}
	\end{equation} 
	The distribution of the range $R_i$ between the target and the visible landmark is detailed in Lemma~\ref{lem1}.
	After taking expectations, the average size of $|A_i|$ is
	\begin{equation}
		\nbbE_{R_i | N_\nbx}\left[|A_i|\right] = \int_{0}^{d_v} s(r) f_D(r) \nrmd r = \frac{\pi \epsilon}{3 d_v^2} \left(8 d_v^3 - 6 d_v^2 \epsilon + \epsilon^3\right).
	\end{equation}	
	This completes the proof.
\end{proof}
Using results from the above lemmas, we present the bounds on the localizability probability based on the ranges.

\begin{theorem}
	\label{the3}
	The localizability probability $P_{\rm Loc}$ based on the ranges is bounded by 
	\begin{align}
		1-\exp\!\left(-2m\right) I_0\!\left(2m\sqrt{\frac{8d_v^3 \epsilon - 6 d_v^2 \epsilon^2 + \epsilon^4}{3d_v^4}}\right)\geq P_{\rm Loc} = 1 - P_{\rm N-Loc} \\
		P_{\rm Loc} = 1 - P_{\rm N-Loc} \geq \exp\!\left(\frac{m^2\epsilon}{3d_v^4}\!\left(8d_v^3 - 6 d_v^2 \epsilon + \epsilon^3\right) - 2m\right),
	\end{align}
	where $m = \lambda \pi d_v^2$ and $I_0(\cdot)$ is the modified Bessel function of the first kind. 
\end{theorem}

\begin{proof}
	Using the same proof strategies in Theorem~\ref{the2} and Lemma~\ref{lem10}, we can write the probability of $\Delta_{\rm s}(F(\nbX_0), F(\nbzero)) \leq \epsilon$ by taking expectations over $\Phi \cap B_\nbx$, given as 
	\begin{align}
		\nbbP\! \left[ \Delta_{\rm s}(F(\nbx), F(\nbzero)) \leq \epsilon  \right] =& \nbbE_{\Phi \cap B_\nbx}\!\left[\nbbP\! \left[ \Delta_{\rm s}(F(\nbx), F(\nbzero)) \leq \epsilon \mid F(\nbx)= f(\nbx) \right]\right] \\
		=& \nbbE_{\Phi \cap B_\nbx}\!\left[\frac{m^{N_\nbx}}{N_\nbx!} e^{-m} \cdot \prod_{i = 1}^{N_\nbx}\left\{\frac{|A_i|}{\pi d_v^2}\right\}\right] \\
		\overset{(a)}{\geq}&  \nbbE_{\Phi \cap B_\nbx}\!\left[\frac{m^{N_\nbx}}{N_\nbx!} e^{-m} \cdot \prod_{i = 1}^{N_\nbx} \nbbE_{R_i \mid N_\nbx}\! \left\{\frac{|A_i|}{\pi d_v^2}\right\}\right] \\
		=& \exp\!\left(-2m\right) \left(I_0\!\left(2\sqrt{\frac{m^2\epsilon}{3d_v^4}\!\left(8d_v^3 - 6 d_v^2 \epsilon + \epsilon^3\right)}\right) - 1\right),
	\end{align}
		where (a) follows from the fact that conditioned on the number of visible landmarks, the ranges $R_i$ are independently and identically distributed, as demonstrated in Lemma~\ref{lem12}. 
		Now, we use the definition of the non-localizability probability in~\eqref{eqdef:C3-P-nloc} and write
		\begin{align}
			P_{\rm N-Loc} &= \nbbP\!\left[N\!\left( B_\nbx\right)  = 0, N\!\left( B_{\nbzero}\right)  = 0 \right] + \nbbP\!\left[\Delta_{\rm s}\!\left(F(\nbx), F(\nbzero)\right) \leq \epsilon \right]\\
			&\geq \exp\!\left(-2m\right) I_0\!\left(2m\sqrt{\frac{8d_v^3 \epsilon - 6 d_v^2 \epsilon^2 + \epsilon^4}{3d_v^4}}\right).\label{eq:c3-lb}
		\end{align}
		The left side of the inequality is proved. With Lemma~\ref{lem11}, the same technique can be used to prove the right side of the inequality. This completes the proof.
\end{proof}
\begin{remark}
	The lower bound of the non-localizability probability in~\eqref{eq:c3-lb} (or equivalently, the upper bound of the localizability probability) based on the ranges has the same closed-form expression as the non-localizability probability based range vectors, presented in~\eqref{thm1:case1}. This result aligns with our intuition that the range vector offers additional information about the ordering of the measurements.
\end{remark}

\section{Results and Discussions}\label{sec:simulation}

We now present a summary of the localizability probability results when employing various types of measurements in Table~\ref{tab:comparison}. These expressions share the same structure $e^{-2m}I_0\left(2\alpha m\right)$, where $\alpha \in \left[0,1\right]$. 
The common component in $\alpha$ is $\frac{\epsilon}{d_v}$, which we interpret as the noise-to-visible-distance ratio.
The expressions of $\alpha$ depend on specific types of measurements, allowing us to conveniently compare the non-localizability probability under different scenarios. 
For instance, the non-localizability probability based on relative locations outperforms that based on ranges, as the corresponding $\alpha$ will be smaller.
\begin{table}[htp]
	\centering
	\caption{The comparison of the non-localizability probabilities under different scenarios.}
	\begin{tabular}{c|c}
		\hline
		\textbf{Measurement type} & \textbf{Non-localizability probability}\\
		\hline\hline
		Number of landmarks & $P_{\rm N-Loc} = e^{-2m} I_0\!\left(2m\right)$\\
		Range vectors & $P_{\rm N-Loc} = e^{-2m} I_0\!\left(2m \cdot \sqrt{\frac{8d_v^3\epsilon - 6 d_v^2 \epsilon^2 + \epsilon^4}{3 d_v^4}}\right)$ \\
		Set of relative locations& $e^{-2m} \exp\!\left(m^2\frac{\epsilon^2}{d_v^2}\right) \geq P_{\rm N-Loc} \geq e^{-2m} I_0\!\left(2m\frac{\epsilon}{d_v}\right)$ \\
		Set of ranges& $e^{-2m} \exp\!\left(m^2\frac{8d_v^3\epsilon - 6 d_v^2 \epsilon^2 + \epsilon^4}{3d_v^4}\right) \geq P_{\rm N-Loc} \geq e^{-2m} I_0\!\left(2m \cdot \sqrt{\frac{8d_v^3\epsilon - 6 d_v^2 \epsilon^2 + \epsilon^4}{3 d_v^4}}\right)$ \\
		\hline
	\end{tabular}
	\label{tab:comparison}
\end{table}

We then verify our analytical findings by comparing the theoretically derived non-localizability probabilities, $P_{\rm N-Loc}$ in Theorems~\ref{the1}, \ref{the2} and~\ref{the3} against the outcomes of Monte-Carlo simulations.
For the simulations, we fix the candidate location at the origin and generate realizations of the PPP. The non-localizability probability is then determined by calculating the frequency at which noisy measurements from two different realizations overlap.
The landmark intensity $m = \lambda \pi d_v^2$ takes various values ranging from $2$ to $8$, providing the average number of landmarks in the visibility region. The maximum visibility distance $d_v$ is fixed at $50$ meters.
Figs.~\ref{fig:c1}, \ref{fig:c2} and \ref{fig:c3} depict the theoretical $P_{\rm N-Loc}$ values from Theorems~\ref{the1}, \ref{the2} and~\ref{the3}, alongside with Monte Carlo simulation results of non-localizability probabilities based on the range vector, the set of relative locations, and the set of ranges, respectively.
As expected, in \figref{fig:c1}, the simulation results perfectly match the theoretical non-localizability probability derived in Theorem~\ref{the1}.
Notably, when the noise-to-visible-distance ratio $\frac{\epsilon}{d_v}$ becomes larger, the result reduces to the one based on the number of visible landmarks, as discussed in Corollary~\ref{cor1}.
All three figures show that $P_{\rm N-Loc}$ increases with $\epsilon$, meaning that the performance degrades when noise increases.
In \figref{fig:c2} and \figref{fig:c3}, when $\epsilon$ and $m$ increase, the regions $A_i$ described in Theorems~\ref{the2} and \ref{the3} are more likely to overlap, which impacts the tightness of our bounds.
Furthermore, as expected by our analysis, $P_{\rm N-Loc}$ decreases with the increment of landmark intensity.
In particular, Proposition~\ref{prop1} establishes that $P_{\rm N-Loc}$ approaches zero as $\lambda$ tends to infinity, indicating the feasibility of error-free localization in this limiting regime.
The non-localizability probabilities derived in this paper provide valuable insights into the design of vision-based positioning systems. 
As an example, consider a finite area of interest with a fixed number of candidate locations, denoted as $n$. 
We can quantify the probability of a particular measurement being unique on the map, given by $p \simeq P_{\rm Loc}^{n-1}$.
For example, with $n=100$ and $P_{\rm Loc} = 0.999$, the result $p\simeq0.905$ means the target can be localized without error with approximate probability $90.5\%$.

\begin{comment}
    \chb{Let $X$ be the number of failures before the first success matching of the measurements from candidate locations that have independent landmark patterns. $X$ is a geometric random variable with mean $E(X) = \frac{1 - P_{\rm N-Loc}}{P_{\rm N-Loc}} \simeq \frac{1}{P_{\rm N-Loc}}$. With $P_{\rm N-Loc} = 10^{-2}$, we may determine the location of the target out of $100$ candidate locations that covers a $100 \times \pi 50^2 {\rm m}^2$.}
\end{comment}

\begin{figure}
	\centering
	\includegraphics[width=0.45\textwidth]{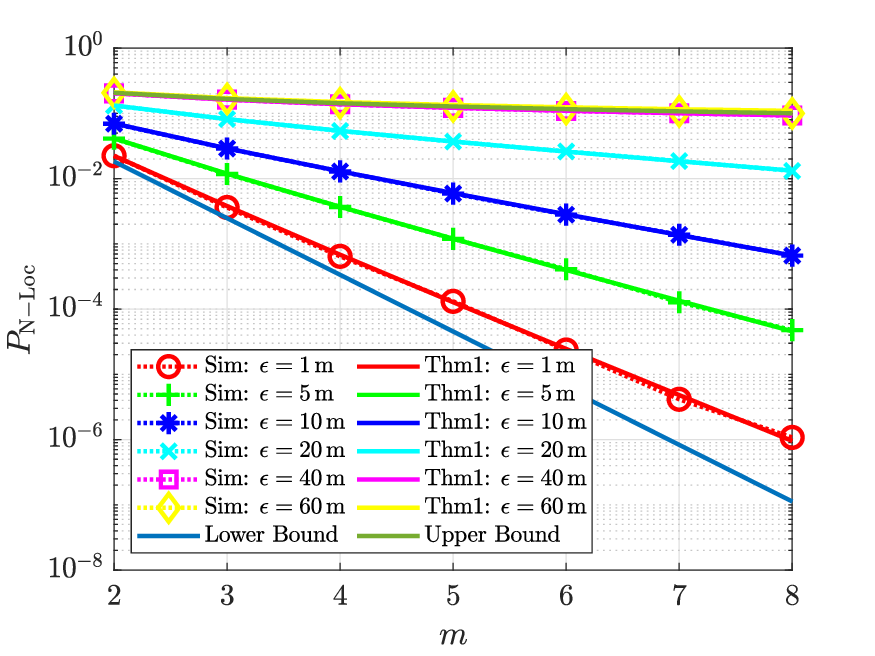}
	\caption{The non-localizability probability based on the range vector (Case 1) as a function of the intensity of landmarks $m$. Markers and solid lines represent simulated and theoretical results, respectively.} 
	\label{fig:c1}
\end{figure}

\begin{figure*}
	\centering
	\subfigure[Upper Bound]{
	\begin{minipage}[t]{0.45\textwidth}
		\centering
		\includegraphics[width=\textwidth]{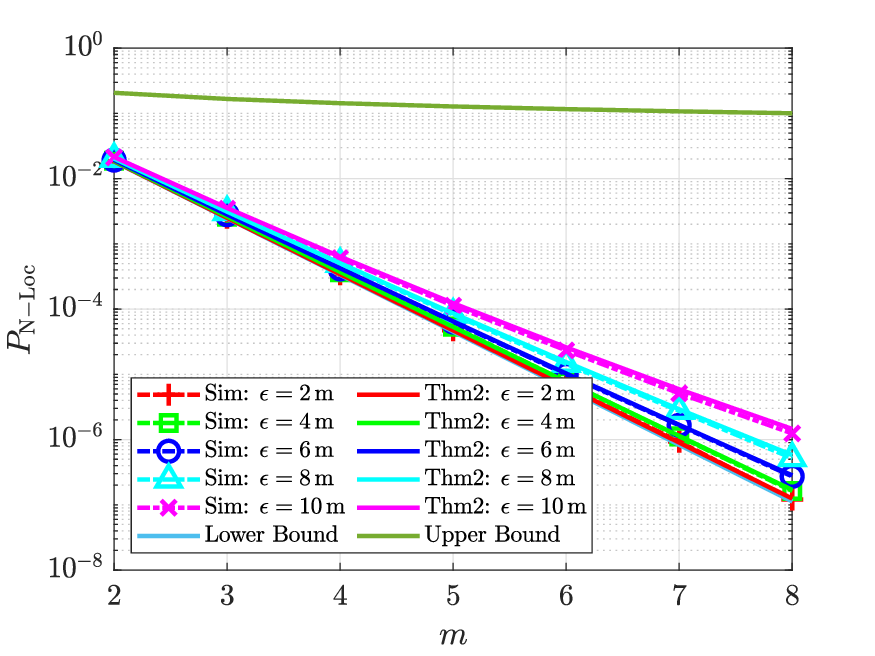}
		\label{fig:c2u}
	\end{minipage}
	}
	\subfigure[Lower Bound]{
	\begin{minipage}[t]{0.45\textwidth}
		\centering
		\includegraphics[width=\textwidth]{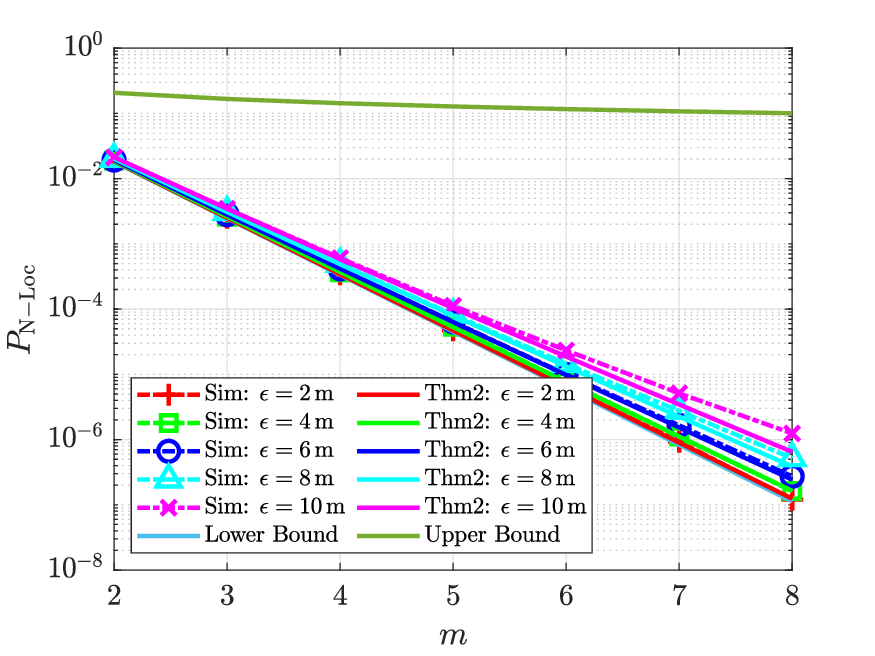}
		\label{fig:c2l}
	\end{minipage}
	}
	\caption{The non-localizability probability based on relative locations (Case 2) as a function of the intensity of landmarks $m$. Markers and solid lines represent simulated and theoretical results, respectively.}
	\label{fig:c2}
\end{figure*}

\begin{figure}
	\centering
	\subfigure[Upper Bound]{
	\begin{minipage}[t]{0.45\textwidth}
		\centering
		\includegraphics[width=\textwidth]{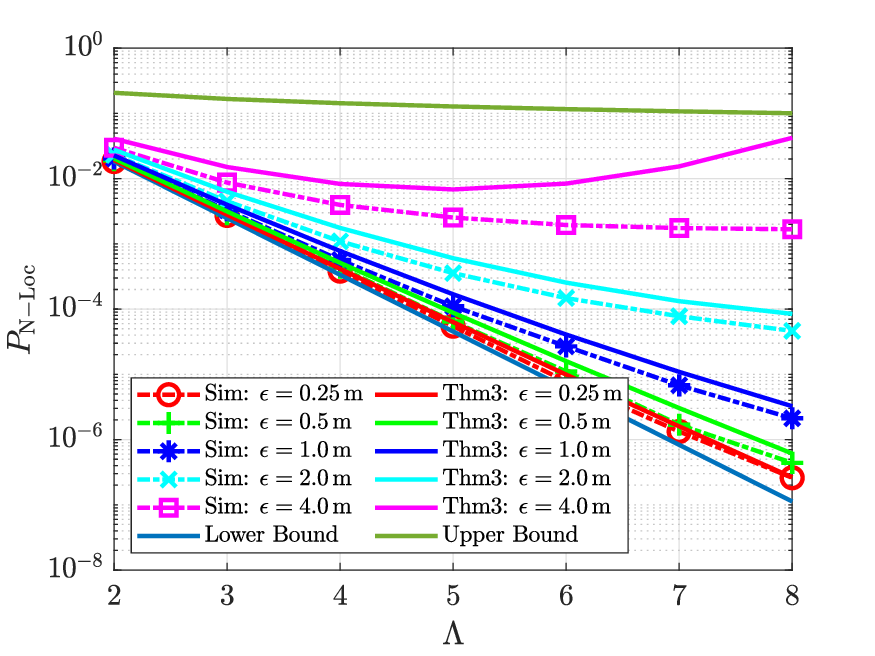}
		\label{fig:c3u}
	\end{minipage}
	}
	\subfigure[Lower Bound]{
	\begin{minipage}[t]{0.45\textwidth}
		\centering
		\includegraphics[width=\textwidth]{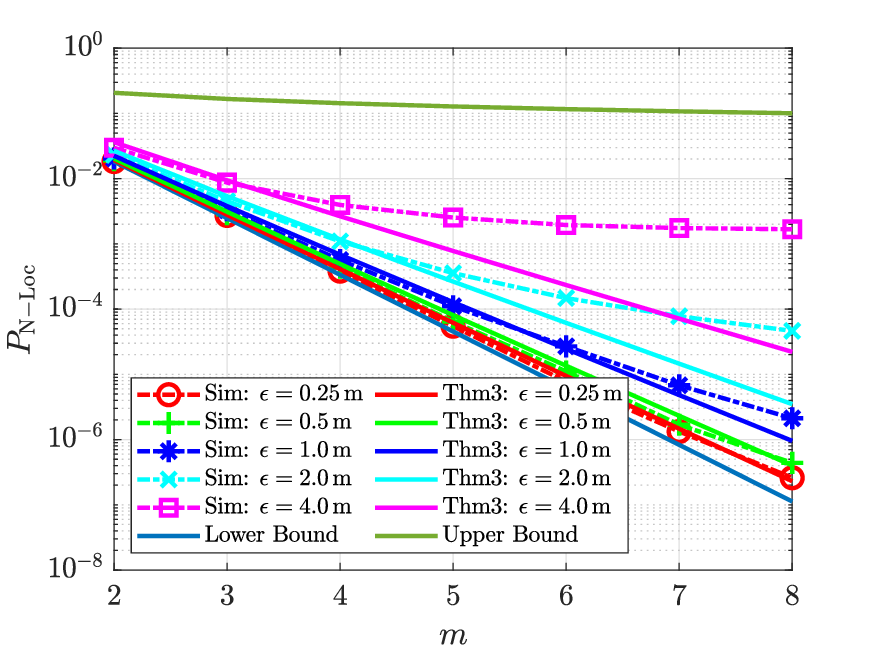}
		\label{fig:c3l}
	\end{minipage}
	}
	\caption{The non-localizability probability based on the set of ranges (Case 3) as a function of the intensity of landmarks $m$. Markers and solid lines represent simulated and theoretical results, respectively.}
	\label{fig:c3}
\end{figure}

\section{Conclusion} \label{sec:conclusion}
In this paper, we have proposed a novel and tractable statistical framework using stochastic geometry.
The concept of {\em localizability} is introduced and rigorously analyzed for vision-based positioning.
In our framework, landmarks are treated as indistinguishable points, and their spatial distribution is modeled using a PPP. 
The target can obtain measurements related to landmarks within the visibility region. 
These measurements can take various forms, including the range vector, the set of relative locations, or even the set of ranges, and are represented as a function of the target location. 
The main technical contributions of this paper are the unified approach to analyzing localizability and the application of this approach to various types of measurements.
One of the key findings is that the localizability probability approaches one as the landmark intensity tends to infinity, analogous to results found in information theory.
Our work provides valuable insights into understanding the limitations and challenges associated with positioning using vision information in the presence of indistinguishable landmarks.

Beyond this work, there are two potential lines for future exploration. 
First, factors such as sensor resolution, environmental conditions, and landmark types can be considered in the information-theoretic analysis of vision-based positioning systems. The foundational limits of localization accuracy in vision-based systems would provide a deeper understanding of existing algorithms. With this understanding, one can develop optimal localization strategies when there are indistinguishable landmarks. For instance, in scenarios where location estimation is inaccurate, the question arises whether it is better to obtain more measurements from indistinguishable landmarks or to identify the existing indistinguishable landmarks.
Additionally, potential connections to codewords in the communication system are also valuable to explore, possibly along the lines of the work in~\cite{anantharam2015capacity} on capacity and error exponents of point processes. In this line of work, the placements of landmarks can be optimized to improve localization accuracy. 

In terms of practical applications, specific algorithms for vision-based positioning that can operate effectively with limited visual information and computational resources are of great interest. 
We have initiated this study in a recent conference presentation~\cite{conf1}, where we used pairwise constraints to identify the correct landmark combination. 
Related to this, while we have demonstrated that error-free localization is achievable as the intensity of landmarks tends to infinity, the measurements considered in specific cases were primarily inspired by the characteristics of the vision sensors and potential deployment scenarios and hence may not necessarily be optimal. 
Therefore, exploring optimal schemes to represent point patterns that ensure the maximum location information is valuable. 
{\em Overall, this paper makes the first attempt to connect stochastic geometry, localization, and computer vision, which could potentially be a new research direction.}

\bibliographystyle{IEEEtran}
\bibliography{citation}

\end{document}